\documentclass{siamart0516}
\usepackage[letterpaper,centering]{geometry}
\usepackage{graphicx}
\usepackage{subcaption}
\usepackage{enumerate}
\usepackage{amsmath,amssymb} 
\usepackage{booktabs}
\usepackage{mathtools} 				% for \refeq, matrix* environments
\usepackage{algorithm}
\usepackage[noend]{algpseudocode}
\usepackage[mathscr]{euscript}
\usepackage[normalem]{ulem} 	% for \emph
\usepackage{pbox} 						% for \pbox
\usepackage{comment}
\usepackage[textsize=tiny]{todonotes}
\usepackage{pgfplots}
\pgfplotsset{compat=1.9}

\newcommand{\pic}[2]{\includegraphics[width = #1\textwidth]{#2}}
\newcommand{\mathup}[1]{\text{\textup{#1}}}
\newcommand{\tx}[1]{\mathrm{#1}}
\newcommand{\lb}{\left}
\newcommand{\mb}{\middle}
\newcommand{\rb}{\right} 

\newcommand{\defeq}{\vcentcolon=}

\newcommand{\Jac}{\nabla}
\newcommand{\Hess}{\nabla^2}
\newcommand{\T}{\top}
\newcommand{\dd}{\mathup{d}}
\newcommand{\N}{\mathup{N}}
\newcommand{\I}{\mathrm{I}} 
\newcommand{\intersect}{\cap}

\newcommand{\R}{\mathbb{R}}
\newcommand{\cH}{\mathcal{H}} 
\newcommand{\cF}{\mathcal{F}} 
\newcommand{\cB}{\mathcal{B}} 
 
\newcommand{\cR}{\mathfrak{R}} 
\newcommand{\uP}{\mathbb{P}} 
\newcommand{\cHcm}{{ \mathcal{H}_\tx{CM} }} 
\newcommand{\barxi}{\eta}

\DeclareMathOperator*{\argmin}{arg\,min}

\makeatletter
\newcommand{\algmargin}{\the\ALG@thistlm}
\makeatother
\newlength{\whilewidth}
\algdef{SE}[parWHILE]{parWhile}{EndparWhile}[1]
  {\parbox[t]{\dimexpr\linewidth-\algmargin}{%
     \hangindent\whilewidth\strut\algorithmicwhile\ #1\ \algorithmicdo\strut}}{\algorithmicend\ \algorithmicwhile}%
\algnewcommand{\parState}[1]{\State%
  \parbox[t]{\dimexpr\linewidth-\algmargin}{\strut #1\strut}}
\newtheorem{assume}[theorem]{Assumption}
\newsiamremark{remark}{Remark}

\title{Scalable optimization-based sampling on function space}

\author{
  Johnathan M.~Bardsley%
  \thanks{Department of Mathematical Sciences, Montana, University of Montana, Missoula, MT 59812 USA (\email{bardsleyj@mso.umt.edu})}
  \and
  Tiangang Cui%
  \thanks{School of Mathematics, Monash University, Victoria 3800, Australia (\email{tiangang.cui@monash.edu})}
  \and
  Youssef M.~Marzouk%
  \thanks{Center for Computational Engineering, Massachusetts Institute of Technology, Cambridge, MA 02139 USA (\email{zheng\_w@mit.edu}, \email{ymarz@mit.edu})}%
  \and
  Zheng Wang%
  \footnotemark[3]
}

\begin{document}

\maketitle

\begin{abstract}
Optimization-based samplers such as randomize-then-optimize (RTO)~\cite{rto}
provide an efficient and parallellizable approach to solving large-scale Bayesian inverse problems. 
These methods solve randomly perturbed optimization problems to draw samples from an approximate posterior distribution. ``Correcting'' these samples, either by Metropolization or importance sampling, enables characterization of the original posterior distribution.
This paper focuses on the scalability of RTO to problems with high- or infinite-dimensional parameters. 
We introduce a new subspace acceleration strategy that makes the computational complexity of RTO scale linearly with the parameter dimension.
This subspace perspective suggests a natural extension of RTO to a function space setting. We thus formalize a function space version of RTO and establish sufficient conditions for it to produce a valid Metropolis--Hastings proposal, yielding \emph{dimension-independent} sampling performance. 
Numerical examples corroborate the dimension-independence of RTO and demonstrate sampling performance that is also robust to small observational noise. 

\end{abstract}


\section{Introduction} \label{sec:introduction}
The Bayesian framework is widely used for uncertainty quantification in inverse problems---i.e., inferring parameters of mathematical models given indirect and noisy data \cite{KaipioSomersalo, StuartActa}.
In a Bayesian setting, the parameters are described as random variables and endowed with prior distributions. Conditioning on an observed data set yields the posterior distribution of these parameters, which characterizes uncertainty in possible parameter values. Solving the inverse problem amounts to computing posterior expectations, e.g., posterior means, variances, marginals, or other summary statistics.

Sampling methods---in particular, Markov chain Monte Carlo (MCMC) algorithms---provide a flexible yet provably convergent way of estimating posterior expectations \cite{MCMC:BGJM_2011}. The design of effective MCMC methods, however, rests on the careful construction of proposal distributions: efficiency demands proposal distributions that reflect the geometry of the posterior \cite{riemann}, e.g., anisotropy, strong correlations, and even non-Gaussianity \cite{mattMCMC}. 
Another significant challenge in applying MCMC %
 is parameter dimensionality. In many inverse problems governed by partial differential equations, the ``parameter'' is in fact a function of space and/or time that, for computational purposes, must be represented in a discretized form. Discretizations that sufficiently resolve the spatial or temporal heterogeneity of this function are often high dimensional.
Yet, as analyzed in \cite{convdifflim,convrates,convrwm,mala}, the performance of many common MCMC algorithms may degrade as the dimension of the discretized parameter increases, meaning that more MCMC iterations are required to obtain an effectively independent sample. One can design MCMC algorithms that do not degrade in this manner by formulating them in function space and ensuring that the proposal distribution satisfies a certain absolute continuity condition~\cite{pCN,StuartActa}. These samplers are called \emph{dimension independent} \cite{pCN,dili}.
Yet another core challenge is that MCMC algorithms are, in general, intrinsically serial: sampling amounts to simulating a discrete-time Markov process. The literature has seen many attempts at parallelizing the evaluation of proposed points \cite{CalderheadPNAS} or sharing information across multiple chains \cite{ConradJUQ,IP:HLH_2003}, but none of these is embarrassingly parallel.

A promising approach to many of these challenges is to \emph{convert optimization methods into samplers} (i.e., Monte Carlo methods). This idea has been proposed in many forms: key examples include randomize-then-optimize (RTO) \cite{rto}, Metropolized randomized-maximum-likelihood (RML) \cite{auRML, MCMC:WBG_2018}, and implicit sampling \cite{generalimplicit,implicit}. 
In its most basic form, RTO requires Bayesian inverse problems with Gaussian priors and noise models, although it can extend to problems with non-Gaussian priors via a change of variables \cite{l1prior}. Metropolized RML has problem requirements similar to those of RTO, but requires evaluating second-order derivatives of the forward model. Implicit sampling applies to target densities whose contours enclose star-convex regions; each proposal sample can then be generated cheaply by solving a line search.

In general, each of these algorithms solves randomly-perturbed realizations of an optimization problem to generate samples from a probability distribution that is ``close to'' the posterior. 
The probability density function of this distribution is computable, and thus the distribution can be used as an independent proposal within MCMC or as a biasing distribution in importance sampling. For non-Gaussian targets, these proposal distributions are non-Gaussian. In general, they are \emph{adapted} to the target distribution. 
The computational complexity and dimension-scalability of the resulting sampler can be linked to the structure of the corresponding optimization problem. 
In addition, these sampling methods are embarrassingly parallel, and are easily implemented with existing optimization tools developed for solving deterministic inverse problems.

This paper considers optimization-based sampling in high dimensions. In particular, we focus on the scalable implementation and analysis of the RTO method. 
To begin, in Section \ref{sec:background} we present interpretations of RTO that provide intuition for the method and its regime of applicability.
Using these interpretations, we next motivate and construct a subspace-accelerated version of RTO whose computational complexity scales linearly with parameter dimension (Section~\ref{sec:scalable}). This approach significantly accelerates RTO in high-dimensional settings. 
Subspace acceleration reveals that RTO's mapping acts differently on different subspaces of the parameter space. 
In Section~\ref{sec:proof}, we exploit this separation of the parameter subspaces to cast the transport map generated by RTO in an infinite-dimensional (i.e., function space) setting \cite{StuartActa}. 
We also establish sufficient conditions for the probability distribution induced by RTO's mapping to be absolutely continuous with respect to the posterior.
This result justifies RTO's observed {\it dimension-independent} sampling behavior: the acceptance rate and autocorrelation time of an MCMC chain using RTO as its proposal do not degrade as the parameter dimension increases. 
Similarly, the performance of importance sampling using RTO as a biasing distribution will stabilize in high dimensions. 
This result is analogous to the arguments in \cite{MCMC:BGLFS_2017,MCMC:BRSV_2008,dili,MCMC:DS_2018,StuartActa} showing that (generalized) preconditioned Crank--Nicolson (pCN), dimension-independent likelihood-informed (DILI) MCMC, and other infinite-dimensional geometric MCMC methods are dimension-independent. 
However, our MCMC construction relies on non-Gaussian proposals in a Metropolis independence setting, where the Markov chain can be run at essentially zero cost \emph{after} the computationally costly step of drawing proposal samples and evaluating the proposal density.
Because the latter step is embarrassingly parallel, the overall MCMC scheme is immediately parallelizable, unlike the above-mentioned MCMC samplers that rely on the iterative construction of Markov chains. 

In Section~\ref{sec:numerics}, we provide a numerical illustration of our algorithm on a one-dimensional elliptic PDE problem, exploring the factors that influence RTO's sampling efficiency. We observe that neither the parameter dimension nor the magnitude of the observational noise influence RTO's performance \emph{per MCMC step}, though they both impact the computational cost of each step. Despite its more costly steps, RTO outperforms simple pCN in this example. 
In Section~\ref{sec:numerics_heat}, we further demonstrate the efficacy of our algorithm on a challenging two-dimensional parabolic PDE problem. 
Overall, our results show that RTO can tackle inverse problems with large parameter dimensions and arbitrarily small observational noise.

\section{Background}
\label{sec:background}

RTO generates samples from an approximation to the target (e.g., posterior) distribution in two steps. First, it repeatedly solves perturbed optimization problems to generate independent proposal samples. Second, it uses this collection of samples to describe an independent proposal for Metropolis--Hastings (MH) or a biasing distribution for self-normalized importance sampling. 
In this section, we first describe the target distributions to which RTO can be applied. %
We then provide interpretations of RTO from the geometric and transport perspectives, which lead to useful insights regarding both the sampling efficiency of RTO and sufficient conditions for the RTO procedure to be valid.
For completeness, we conclude this section by summarizing RTO and other comparable optimization-based sampling methods using the transport map interpretation.

\subsection{Target distribution} 
\label{sec:target}
RTO applies to target distributions on $\mathbb{R}^n$ whose densities can be written as
\begin{equation}  
\pi_\tx{tar}(v) \propto \exp\lb(-\frac12 \lb\| H(v) \rb\| ^2\rb), \label{eq:lsform}
\end{equation}
where $H: \R^n \to \R^{m+n}$ is a vector-valued function of the parameters $v \in \mathbb{R}^n$ with an output dimension of $n+m$, for any $m \geq 1$. This structure is found in Bayesian inverse problems and other similar problems with $n$ parameters, $m$ observations, a Gaussian prior, and additive Gaussian observational noise. To illustrate, let
\[
y = F(u) + \epsilon, \quad \epsilon \sim \N(0,\Gamma_\tx{obs}), \quad u \sim \N(m_\tx{pr},\Gamma_\tx{pr}),
\]
where $y \in \R^m$ is the data, $F:\R^n \to \R^m$ is the forward model, $u \in \R^n$ is the unknown parameter, and $\epsilon \in \R^m$ is the additive noise, assumed independent of $u$. Here, $m_\tx{pr}$ is the prior mean, and $\Gamma_\tx{obs}$ and $\Gamma_\tx{pr}$ are the covariance matrices of the observation noise and prior.
We can simplify the problem via an affine change of variables that transforms the covariance matrices to identity matrices. 
Defining matrix factorizations of the covariances of prior and observation noise
\[
S^{}_\tx{pr} S_\tx{pr}^\T \defeq \Gamma_\tx{pr}, 	\quad S^{}_\tx{obs} S_\tx{obs}^\T \defeq \Gamma_\tx{obs},
\]
we have new whitened variables,
\[ 
v \defeq S_\tx{pr}^{-1} (u - m_\tx{pr}), \quad  G(v) \defeq S_\tx{obs}^{-1}\lb[ F\lb(S_\tx{pr}^{~} v + m_\tx{pr}\rb) - y \rb], \quad  e \defeq S_\tx{obs}^{-1}\epsilon,
\]
where $v \in \R^n$ is the whitened unknown parameter, $G:\R^n \to \R^m$ is the whitened forward model, and $e$ is the whitened observational noise. The inverse problem becomes
\[
0 = G(v) + e, \quad e \sim \N(0,\I), \quad v \sim \N(0,\I) .
\]
The data is shifted to the origin after whitening. The posterior density of the whitened variable $v$ is then
\[ p(v\vert y) = \pi_\tx{tar}(v) \propto \exp\lb(-\frac12 \lb\|\begin{bmatrix} v \\ G(v) \end{bmatrix}\rb\|^2 \rb) ,
\]
which is in the required form (\ref{eq:lsform}) with $H$ defined as
\begin{equation}\label{eq:H} H(v) \defeq \begin{bmatrix} v \\ G(v) \end{bmatrix}. \end{equation}
Given a sample $v$ from the target density $\pi_\tx{tar}(v)$, we can obtain a posterior sample of the original parameter $u$ by applying the transformation
\[ 
u = S_\tx{pr} v + m_\tx{pr}.
\]

Notice that the form of $\pi_\tx{tar}(v)$ in \eqref{eq:lsform} is identical to the probability density function of an $(n+m)$--dimensional standard normal distribution, $\pi(w) \propto \exp\lb( -\frac12 \|w\|^2 \rb)$, evaluated at $w = H(v)$. 
This paints the geometric picture of the required target distribution: the target density $\pi_\tx{tar}(v)$, up to a normalizing constant, is the same as the density of the $(n+m)$--dimensional Gaussian distribution evaluated on the $n$--dimensional manifold $H(v)=(v,G(v)) \subset \R^{m+n}$ parameterized by $v \in \R^n$. 

\subsection{The RTO algorithm}\label{sec:geo_rto}
\newcommand{\vpi}{v_\tx{prop}^{(i)}}
The RTO algorithm requires an orthonormal basis for an $n$-dimensional subspace of $\R^{n+m}$. 
Let this basis be collected in a matrix $Q \in \R^{(m+n)\times n}$ with orthonormal columns. One common choice of $Q$ follows from first finding a linearization point $v_\tx{ref}$, which is often (but not necessarily) taken to be the maximum of the target density, i.e., 
\begin{equation} \label{eq:opmode} v_\tx{ref} = \argmin_{v} \frac12 \lb\|  H(v) \rb\| ^2.\end{equation}
Then one can compute $Q$ from a thin QR factorization of $\nabla H(v_\tx{ref})$; this sets the basis to span the range of $\nabla H(v_\tx{ref})$.

Using this matrix, RTO obtains proposal samples $\vpi$ by repeatedly drawing independent $(n+m)$--dimensional standard normal vectors $\barxi^{(i)}$ and solving the nonlinear system of equations
\begin{equation}\label{eq:rtoequation}
Q^\T H(\vpi) = Q^\T \barxi^{(i)} ,
\end{equation}
which is equivalent to solving the optimization problem
\begin{equation} \label{eq:optprob}
\vpi = \argmin_{v} \frac12 \lb\|  Q^\T\lb( H(v) - \barxi^{(i)}\rb) \rb\| ^2,  
\end{equation}
\emph{if} the minimum of the objective function in \eqref{eq:optprob} is zero.
To ensure that the system of equations \eqref{eq:rtoequation} has a unique solution and that the probability density of the resulting samples can be calculated explicitly, RTO requires the following conditions \cite{rto}. 
\begin{assume}[Sufficient conditions for valid RTO]
\label{assumoverall}
\begin{enumerate}
	\item The function $H$ is continuously differentiable with Jacobian $\nabla H$.
	\item The Jacobian $\nabla H(v)$ has full column rank for every $v$.
	\item The map $v \mapsto Q^\T H(v)$ is invertible. \label{assum3}
\end{enumerate}
\end{assume}

\begin{figure}[h]
\centering
\begin{subfigure}{0.49\textwidth}
\centering
\pic{1}{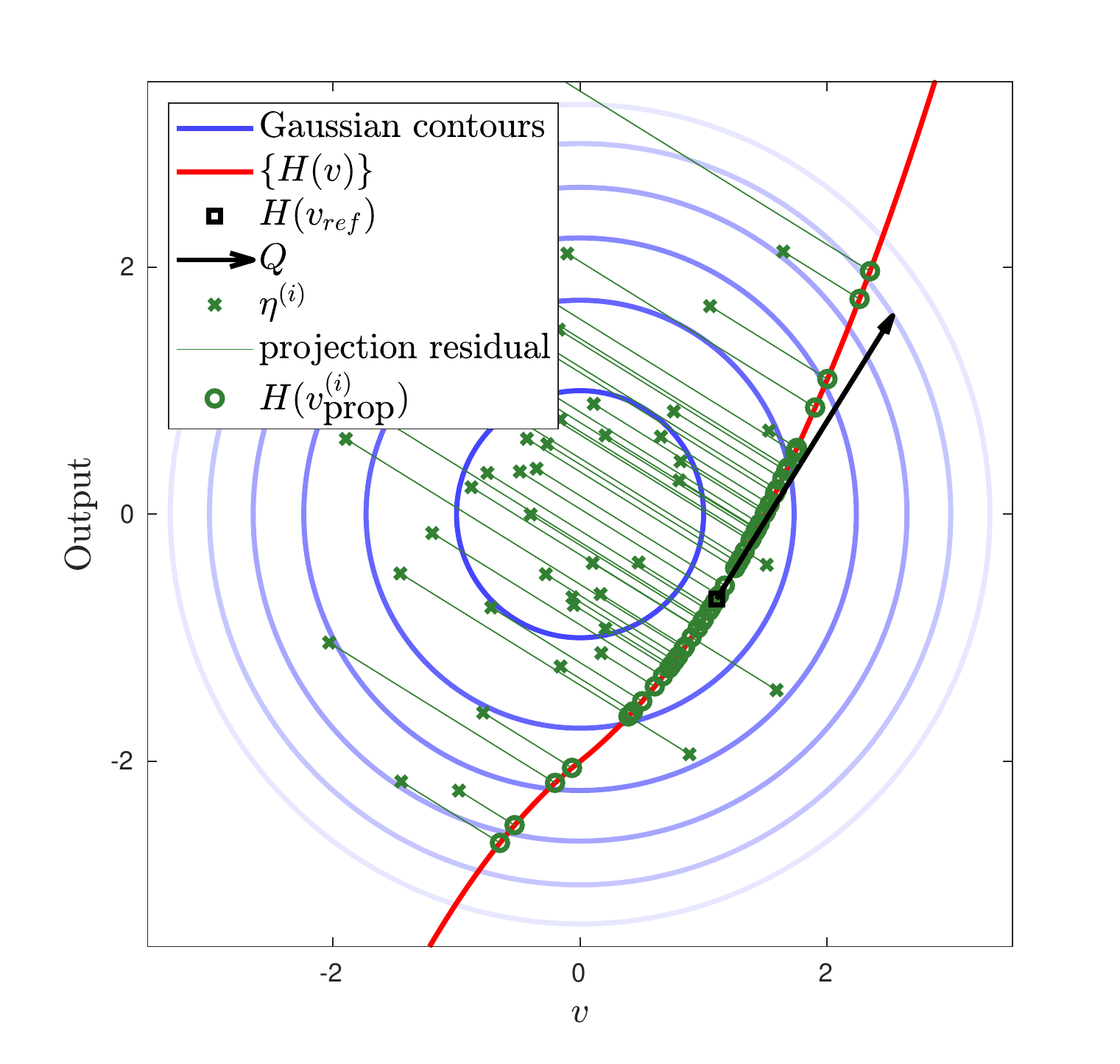}
\caption{Top view}
\end{subfigure}
\begin{subfigure}[t]{0.49\textwidth}
\centering
\pic{1}{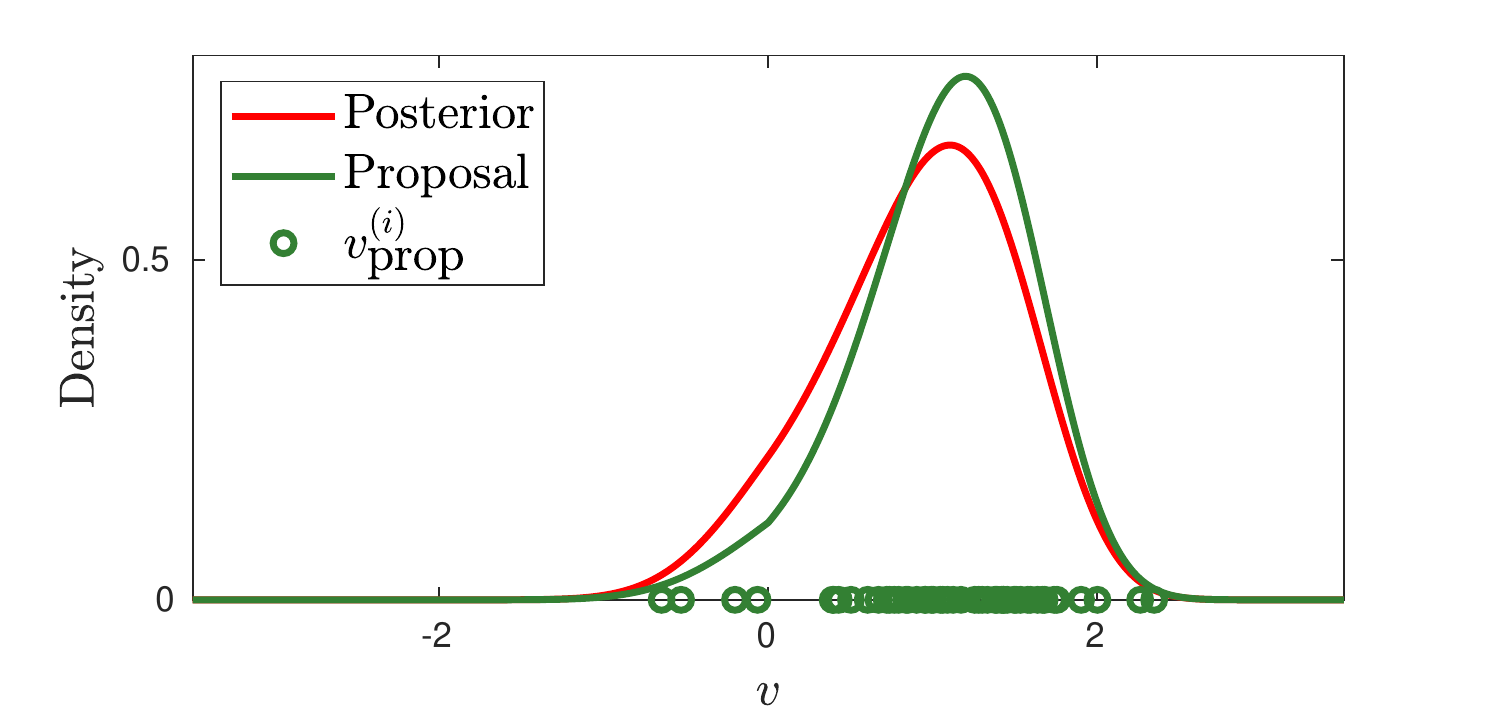}
\caption{Front view}
\end{subfigure}
\caption{Geometric interpretation of RTO, in the case $n=m=1$. RTO projects the $(n+m)$--dimensional Gaussian samples $\barxi$ (green crosses) onto the manifold $\{H(v)\}$ (red line) to determine the proposal samples $v_\tx{prop}$ (green circles). 
The projection residual $H(\vpi) - \barxi$ is orthogonal to the range of $Q$.
In the front view, the proposal samples $v_\tx{prop}$ (green circles) are shown to be distributed according to a proposal density (green line) that is close to the target density (red line).}\label{fig:proj}
\end{figure}

Proposal samples generated via RTO can be interpreted as a projection of $(n+m)$--dimensional Gaussian samples onto the $n$--dimensional manifold $\lb\{ H (v) \, : \, v \in \R^n \rb\}$. The samples are projected along the directions orthogonal to the range of $Q$ such that the condition in \eqref{eq:rtoequation} is satisfied. 
Figure~\ref{fig:proj} depicts the steps of RTO's proposal for the case $n=m=1$. 
This geometric interpretation also illustrates the importance of the third condition for the RTO procedure to be valid: there will be a unique projected vector on the manifold $\{H(v)\}$ for any given $\barxi \sim \N(0,\I_{n+m})$ provided the map $v \mapsto Q^\T H(v)$ is invertible. %

The projection defined by RTO realizes the action of a particular transport map. 
Since the random vector $\barxi \in \R^{n+m}$ is a standard Gaussian and the columns of $Q$ are orthonormal, the projection of $\barxi$, denoted by $\xi \defeq Q^\T \barxi  \in \R^n$, is also a standard Gaussian.
Writing the left hand side of \eqref{eq:rtoequation} compactly as
\[S(\cdot) = Q^\T H(\cdot),\]
the nonlinear system of equations in \eqref{eq:rtoequation} can be expressed as
\begin{equation}\label{eq:RTOmap} S(v) = \xi , \quad \textrm{where\quad}\xi \sim \N(0,\I_n).\end{equation}
This equation describes a deterministic coupling between the target random variable $v \in \R^n$ and the standard Gaussian ``reference'' random variable $\xi \in \R^n$. 
The coupling is defined by the forward model, the data, the observational noise, and the prior, through the function $H$ and the matrix $Q$.

Under the conditions in Assumption \ref{assumoverall}, solving the nonlinear system \eqref{eq:RTOmap} implicitly inverts the transport map $S$; that is, it evaluates $S^{-1}$ on each $\xi$, to obtain a proposal $v = S^{-1}(\xi)$. 
The normalized probability density of $v$ generated by RTO is given by the pushforward density of the $n$-dimensional standard Gaussian under the mapping $S^{-1}$:
\begin{align}
\pi_\tx{RTO}(v) %
&= \lb| \det \Jac S(v) \rb| \pi_\tx{ref}\lb(S(v)\rb)\nonumber\\
 &= (2\pi)^{-\frac n2} \lb| \det \lb( Q^\T \Jac H(v) \rb) \rb| \exp\lb(- \frac12\lb\|Q^\T H (v)\rb\|^2\rb)\label{eq:RTOproposal}.
\end{align}

As shown in \cite{rto}, %
RTO's proposal is exact (i.e., is the target) when the forward model is linear, and its proposal is expected to be close to the target when the forward model is close to linear. For weakly nonlinear problems, the proposal can be a good approximation to the posterior and hence can be used in MCMC and importance sampling.
These proposal samples can be used either as an independent proposal in Metropolis--Hastings (MH) or as a biasing distribution in importance sampling. For the former case, the Metropolis--Hastings acceptance ratio can be written as
\[
\frac{\pi_\tx{tar}(\vpi)\,\pi_\tx{RTO}(v^{(i-1)})}{\pi_\tx{tar}(v^{(i-1)})\,\pi_\tx{RTO}(\vpi)} =\frac{w(\vpi)}{w(v^{(i-1)})},
\]
where the weight $w(v)$ is defined as
\begin{equation}\label{eq:w(theta)}
w(v) 
=  \lb| \det\lb( Q^\T \Jac{H}(v)\rb) \rb| ^{-1} \exp\lb(-\frac12 \lb\| H(v)\rb\| ^2+\frac12 \lb\| Q^\T H(v) \rb\| ^2\rb).
\end{equation}
The resulting method, called RTO--MH, is summarized in Algorithm~\ref{alg:rto}.

For importance sampling, since the normalizing constant of the target density is unknown, the weights must be normalized as
\[
\tilde w(v^{(i)}) = w(v^{(i)})/\sum_{j=1}^{N} w(v^{(j)}),
\] 
where $N$ is the number of samples and the sum of weights $\tilde w(v^{(i)})$ is thus one.
The proposal samples and weights can then be used to compute posterior expectations of some quantity of interest $g(v)$ using the self-normalizing importance sampling formula: 
\[ \int g(v) \pi_\tx{tar} (v) \dd v = \sum_{i=1}^{N} \tilde w(\vpi) g(\vpi) . \]

\begin{algorithm}[h!t]
\caption{RTO Metropolis--Hastings (RTO-MH)} \label{alg:rto}
\begin{algorithmic}[1]
\State Find $v_\tx{ref}$ using \eqref{eq:opmode}
\State Determine $\Jac{H}(v_\tx{ref})$
\State Compute $Q$, whose columns are an orthonormal basis for the range of $\Jac{H}(v_\tx{ref})$
\For {$i = 1, \ldots, n_\tx{samps}$} in parallel
  \State Sample $\barxi^{(i)}$ from an $(n+m)$--dimensional standard normal distribution
  \State Solve for a proposal sample $\vpi$ using \eqref{eq:optprob}
  \State Compute $w(\vpi)$ from \eqref{eq:w(theta)}
\EndFor

\State Set $v^{(0)} = v_\tx{ref}$
\For {$i = 1, \ldots, n_\tx{samps}$} in series
\State Sample $t$ from a uniform distribution on [0,1]
\If {$t < \lb. w(\vpi) \rb/ w(v^{(i-1)})$}
  \State $v^{(i)}$ = $\vpi$
\Else
  \State $v^{(i)}$ = $v^{(i-1)}$
\EndIf
\EndFor
\end{algorithmic}
\end{algorithm}

\begin{remark}
For Bayesian inverse problems, the RTO formulation presented here is limited to cases with Gaussian prior and Gaussian observation noise. By transforming non-Gaussian prior densities and/or observation noises into Gaussian densities, this limitation may be relaxed. See \cite{l1prior,dim_robust} for examples.
\end{remark}

Similarly to RTO, other optimization-based samplers such as random-map implicit sampling \cite{implicit} and Metropolized RML \cite{auRML} also use a standard normal as the reference distribution and push forward this distribution through some deterministic transformation. Each of these samplers specifies a different inverse transport $S$, as in \eqref{eq:RTOmap}, and then solves an optimization problem to evaluate $S^{-1}$ on each reference sample.
For all three algorithms, the pushforward of the reference distribution can be used as a proposal distribution in Metropolis--Hastings or as a biasing distribution in importance sampling. 
A summary of each algorithm's mapping is given in Appendix \ref{sec:app_imp}.
The subspace acceleration strategies and infinite-dimensional formulation of RTO developed in this work may also benefit implicit sampling and RML. 
In addition, interpreting optimization-based samplers as transport maps and utilizing the importance sampling formula naturally open the door to constructing multilevel \cite{multilevelpath,multilevel} and multi-fidelity \cite{multifidelity} Monte Carlo estimators for Bayesian computation, enabling additional speedups.
Further research along this direction is in \cite{chen_2019}.

\section{Scalable implementation of RTO}\label{sec:scalable}
In high-dimensional problems, the computation cost of operations involving the dense matrix $Q$ in RTO poses a major computational challenge:
The matrix-vector product with $Q^\T$ in each evaluation of the objective function in \eqref{eq:optprob} costs $\mathcal{O}((n+m)\times n)$ floating point operations, where $n$ is the number of parameters and $m$ is the number of observed data. 
Assembling the matrix $Q^\T \nabla H(v)$ also requires $n + m$ matrix-vector products, and an additional $\mathcal{O}(n^3)$ floating point operations are needed to compute the determinant in the proposal density \eqref{eq:RTOproposal}. 
For high-dimensional parameters, these operations are computationally prohibitive to apply. 
To overcome this challenge, we introduce a new subspace acceleration strategy to make these RTO operations scale linearly with the parameter dimension.

\subsection{Subspace acceleration}
Our scalable implementation avoids computing and storing the QR factorization of the full-rank $(n+m) \times n$ matrix $\Jac H(v_\tx{ref})$. Instead, it opts to construct (and store) a singular value decomposition (SVD) of the smaller $m\times n$ linearized forward model $\Jac G(v_\tx{ref})$. 
To begin, we note from the definition \eqref{eq:H} of $H$ that
\[ 
\Jac H(v) = \begin{bmatrix}\I \\ \Jac G(v)\end{bmatrix}.
\]
Recall from RTO's mapping \eqref{eq:RTOmap} that the RTO proposal samples are found by
\begin{align*} 
Q^\T H(v) = \xi , \quad \textrm{where\quad}\xi \sim \N(0,\I_n), 
\end{align*}
where the columns of $Q$ form an orthonormal basis for the range of $\Jac H(v_\tx{ref})$ and $Q$ is computed from the thin QR decomposition of $\Jac H(v_\tx{ref})$. 
Since the 2-norm used in the objective function \eqref{eq:optprob}, the determinant in \eqref{eq:RTOproposal}, and the standard Gaussian used in the RTO's mapping \eqref{eq:RTOmap} are all invariant up to a rotation defined by an orthogonal matrix, any orthonormal basis for the range of $\Jac H(v_\tx{ref})$ plays the same role in RTO. 
This offer a viable way to avoiding computing the dense $(m+n)\times n$ matrix $Q$.

Instead of computing the QR factorization of $\nabla H$, we consider the polar decomposition \cite{golub}:
\[
\Jac H(v_\tx{ref}) = \widetilde Q \, \left(\Jac H(v_\tx{ref}) ^\T \Jac H(v_\tx{ref}) \right)^{1/2},
\]
where the matrix $\widetilde Q \in \R^{(m+n)\times n}$ has orthonormal columns and the matrix $\left(\Jac H(v_\tx{ref}) ^\T \Jac H(v_\tx{ref}) \right)^{1/2} \in \R^{n\times n}$ is positive definite by construction. This way, the matrix $\widetilde Q$ can be constructed as
\[
\widetilde Q = \Jac H(v_\tx{ref}) \left(\Jac H(v_\tx{ref}) ^\T \Jac H(v_\tx{ref}) \right)^{-\frac12}.
\]
In the above equation, the matrix $\Jac H(v_\tx{ref}) ^\T \Jac H(v_\tx{ref})$ is the Gauss-Newton approximation of the Hessian of the log-posterior density (referred to as Gauss-Newton Hessian hereafter) defined at the reference parameter $v_\tx{ref}$.

\begin{proposition}
Let $\Jac G(v_\tx{ref})$ denote the forward model linearized at parameter $v_\tx{ref}$, and consider its reduced SVD, 
\[
\Jac G(v_\tx{ref}) = \Psi \Lambda \Phi^\T .
\] 
The nonlinear system $\widetilde Q^\T H(v) = \xi$ defining the RTO mapping can be rewritten as
\begin{equation} \label{eq:scale}
\left\{
\begin{aligned}
(\I_n - \Phi\Phi^\T)\,\xi  &= (\I_n - \Phi\Phi^\T)\,v \\
\Phi\Phi^\T\,\xi &= \Phi\,\Big[(\Lambda^2 + \I_r)^{-\frac12} \big( \Phi^\T  v + \Lambda \Psi^\T G(v) \big)\Big].
\end{aligned}
\right.\end{equation}
The weighting function $w(v)$ in \eqref{eq:w(theta)} can be expressed as
\begin{equation}\label{eq:w_svd}
w(v) =  \lb| \det\lb( \widetilde Q^\T \Jac{H}(v)\rb) \rb| ^{-1} \exp\lb(-\frac12 \lb\| G(v)\rb\| ^2 - \frac12\lb\|\Phi^\T  v\rb\|^2 +\frac12 \lb\| (\Lambda^2 + \I_r)^{-\frac12} \big( \Phi^\T  v + \Lambda \Psi^\T G(v) \big) \rb\| ^2\rb),
\end{equation}
where the determinant takes the simplified form
\begin{equation}
\lb| \det\lb(\widetilde Q^\T \nabla H(v) \rb)\rb| = \lb| \det (\Lambda^2 + \I_r)^{-\frac12} \rb| \lb| \det\lb( \I_r + \Lambda \Psi^\T \nabla G(v) \Phi \rb)\rb|\label{eq:det}.
\end{equation}
\end{proposition}

\begin{proof}
We will show that the matrices in the polar decomposition of $\Jac H(v_\tx{ref})$ can be obtained from the reduced SVD $\Jac G(v_\tx{ref}) = \Psi \Lambda \Phi^\T$. 
The eigendecomposition of the Gauss-Newton Hessian can be written in terms of the reduced SVD as
\begin{equation}
\Jac H(v_\tx{ref}) ^\T \Jac H(v_\tx{ref}) = \Phi(\Lambda^2 + \I_r)\Phi^\T + (\I_n - \Phi\Phi^\T),\label{eq:eig}
\end{equation}
where $\I_r$ and $\I_n$ are the identity matrices of size $r\times r$ and $n\times n$, respectively.
Then, we have the identity
\[ 
\left(\Jac H(v_\tx{ref}) ^\T \Jac H(v_\tx{ref}) \right)^{-\frac12} = \Phi(\Lambda^2 + \I_r)^{-\frac12}\Phi^\T + (\I_n - \Phi\Phi^\T).
\]
After some algebraic manipulation, this leads to the matrix 
\[
\widetilde Q = \begin{bmatrix} \Phi(\Lambda^2 + \I_r)^{-\frac12}\Phi^\T + (\I_n - \Phi\Phi^\T) \\ \Psi\Lambda(\Lambda^2 + \I_r)^{-\frac12}\Phi^\T \end{bmatrix},
\]
and hence we have
\begin{align}
\widetilde Q^\T H(v) &= \left[\Phi(\Lambda^2 + \I_r)^{-\frac12}\Phi^\T + (\I_n - \Phi\Phi^\T) \right] v + \Phi\Lambda(\Lambda^2 + \I_r)^{-\frac12}\Psi^\T G(v)\nonumber\\
&= \Phi \left[(\Lambda^2 + \I_r)^{-\frac12} \big(\Phi^\T  v + \Lambda\Psi^\T G(v) \big)\right] + (\I_n - \Phi\Phi^\T)\, v.\label{eq:svdet1}
\end{align}
Thus the nonlinear system $\widetilde Q^\T H(v) = \xi$ can be rewritten in the form \eqref{eq:scale}.
Replacing $Q$ with $\widetilde Q$ in the weighting function $w(v)$ \eqref{eq:w(theta)}, we have
\[
w(v) =  \lb| \det\lb(\widetilde Q^\T \Jac{H}(v)\rb) \rb| ^{-1} \exp\lb(-\frac12 \lb\| H(v)\rb\| ^2+\frac12 \lb\|\widetilde Q^\T H(v) \rb\| ^2\rb).
\]
Since the matrix $\Phi$ has orthonormal columns, $\Phi\Phi^\T$ and $\I_n - \Phi\Phi^\T$ are orthogonal projectors. This leads to the identities
\begin{align*}
\lb\| H(v)\rb\| ^2 & = \lb\| v\rb\| ^2 + \lb\| G(v)\rb\| ^2 = \lb\| (\I_n - \Phi\Phi^\T) v\rb\| ^2 + \lb\| \Phi^\T v\rb\| ^2 + \lb\| G(v)\rb\| ^2 ,\\
\lb\| \widetilde Q^\T H(v) \rb\|^2 & = \lb\| (\I_n - \Phi\Phi^\T) v\rb\| ^2 + \lb\| (\Lambda^2 + \I_r)^{-\frac12} \big(\Phi^\T  v + \Lambda\Psi^\T G(v) \big) \rb\|^2,
\end{align*}
by the definition of $H(v)$ in \eqref{eq:H} and the definition of $\widetilde Q^\T H(v)$ in \eqref{eq:svdet1}.
Substituting the above identities into $w(v)$, we obtain the result in \eqref{eq:w_svd}.
Using (\ref{eq:svdet1}), we obtain the linearization
\[
\widetilde Q^\T \nabla H(v) = \I + \Phi\lb[ (\Lambda^2 + \I_r)^{-\frac12} \Phi^\T  - \Phi^\T + \Lambda (\Lambda^2 + \I_r)^{-\frac12} \Psi^\T \nabla G(v)\rb]. 
\]
Hence, the determinant term is given by
\begin{align*}
\lb| \det\lb(\widetilde Q^\T \nabla H(v) \rb)\rb|
	&= \lb| \det \lb(\I + \Phi\lb[ (\Lambda^2 + \I_r)^{-\frac12} \Phi^\T  - \Phi^\T + \Lambda (\Lambda^2 + \I_r)^{-\frac12} \Psi^\T \nabla G(v)\rb] \rb) \rb|\nonumber\\
	&= \lb| \det \lb(\I_r + \lb[ (\Lambda^2 + \I_r)^{-\frac12} \Phi^\T  - \Phi^\T + \Lambda (\Lambda^2 + \I_r)^{-\frac12} \Psi^\T \nabla G(v)\rb] \Phi \rb)\rb|\nonumber\\
	&= \lb| \det (\Lambda^2 + \I_r)^{-\frac12} \rb| \lb| \det\lb( \I_r + \Lambda \Psi^\T \nabla G(v) \Phi \rb)\rb|,
\end{align*}
where in the second line above we use Sylvester's determinant identity.
This concludes the proof.
\end{proof}

\begin{remark}
For high dimensional problems, it is not feasible to explicitly construct the linearized forward model  $\nabla G(v)$. Instead, one should use matrix-free solvers such as Lanczos or randomized SVD (see \cite{golub,halko} and references therein) to compute the SVD of $\nabla G(v)$. This only involves evaluating matrix-vector products (MVPs) with $\nabla G(v)$ and its adjoint.
\end{remark}

Equation (\ref{eq:scale}) separates $\xi$ into two parts: one in the column space of $\Phi$ and another in its orthogonal complement. 
Defining
\[
v_r = \Phi^\T v, \quad \tx{and} \quad v = \Phi v_r + v_\perp,
\]
where $v_\perp$ is an element in the orthogonal complement of $\tx{range}(\Phi)$, we can solve the nonlinear system of equations \eqref{eq:scale} by first computing
\begin{equation}\label{eq:svd1}
v_\perp = (\I_n - \Phi\Phi^\T)\,\xi, \end{equation}
and then solving the $r$--dimensional optimization problem
\begin{equation}\label{eq:svd2}
 v_r = \argmin_{v_r'} \left\| (\Lambda^2 + \I_r)^{-\frac12} \big( v_r' + \Lambda \Psi^\T G\left( v_\perp + \Phi v_r'\right) - \Phi^\T \xi \big)\right\|^2. \end{equation}
Equations (\ref{eq:svd1}) and (\ref{eq:svd2}) replace the $n$--dimensional optimization problem in \eqref{eq:optprob}.
Note that at each given $v = v_\perp + \Phi v_r$, the vector-valued function within the 2-norm in \eqref{eq:svd2} has the linearization
\begin{equation}\label{eq:svd3}
(\Lambda^2 + \I_r)^{-\frac12} \big( \I_r + \Lambda\Psi^\T \, \nabla G(v)\Phi\big),
\end{equation}
w.r.t.~the reduced-dimensional parameter $v_r$. 
MVPs with the linearization \eqref{eq:svd3} and its adjoint are needed by nonlinear optimization algorithms, e.g., quasi-Newton with line search or trust region with inexact Newton--CG \cite{nocedal}, to solve \eqref{eq:svd2}.
The scalable implementation of RTO is outlined in Algorithm~\ref{alg:scale}.

\begin{algorithm}[htbp]
\caption{Scalable implementation of RTO--MH} \label{alg:scale}{}
\begin{algorithmic}[1]
\State Find $v_\tx{ref}$ using \eqref{eq:opmode}.
\State Determine the Jacobian matrix of the forward model, $\Jac{G}(v_\tx{ref})$.
\State Compute $\Psi$, $\Lambda$ and $\Phi$, which is the SVD of $\Jac{G}(v_\tx{ref})$.
\For {$i = 1, \ldots, n_\tx{samps}$} in parallel
  \State Sample $\xi^{(i)}$ from an $n$--dimensional standard normal distribution.
  \State Solve for a proposal sample $\vpi = v_\perp + \Phi v_r$ using \eqref{eq:svd1} and \eqref{eq:svd2}.
  \State Compute $w(\vpi)$ from \eqref{eq:w_svd} using the determinant from \eqref{eq:det}.
\EndFor

\State Set $v^{(0)} = v_\tx{ref}$.
\For {$i = 1, \ldots, n_\tx{samps}$} in series
\State Sample $t$ from a uniform distribution on [0,1].
\If {$t < \lb. w(\vpi) \rb/ w(v^{(i-1)})$}
  \State $v^{(i)}$ = $\vpi$.
\Else
  \State $v^{(i)}$ = $v^{(i-1)}$.
\EndIf
\EndFor
\end{algorithmic}
\end{algorithm}

\subsection{Computational complexity and rank truncation}\label{sec:rto_complexity}
The computational cost of the scalable RTO implementation derived above has two major sources. \emph{First}, 
producing each RTO sample requires several optimization iterations. Each optimization iteration may evaluate the RTO objective function in \eqref{eq:svd2}, and MVPs with the linearization \eqref{eq:svd3} and its adjoint, several times. These operations require evaluating the forward model, the actions of the linearized forward model and its adjoint, and the actions of the matrices $\Phi$ and $\Psi$ several times. \emph{Second}, for each RTO sample, we need to evaluate the determinant in \eqref{eq:det} once to compute the weighting function. This in turn involves evaluating $r$ MVPs with $\nabla G(v)$ and computing the determinant of a $r \times r$ matrix. 

The following proposition summarizes the computational complexity of the operations involved in  Algorithm~\ref{alg:scale}.

\begin{proposition}\label{prop:complexity}
We adopt the following assumptions on the computation of each RTO sample to establish the computational complexity of the scalable implementation of RTO.
\begin{enumerate}
	\item On average, $k_\tx{opt}$ optimization iterations are needed to obtain each RTO sample. On average, $k_\tx{obj}$ objective function evaluations and $k_\tx{adj}$ MVPs with the linearization \eqref{eq:svd3} and its adjoint are needed within each optimization iteration. 
	\item The number of floating  point operations required to evaluate the forward model $G(v)$ is a function of the dimension of the discretized parameters, denoted by $C_1(n)$.
	\item The number of floating point operations required to compute an MVP with the linearized forward model $\nabla G(v)$ and its adjoint is a function of the dimension of the discretized parameters, denoted by $C_2(n)$.
	\item The data dimension is less than the parameter dimension, i.e., $m < n$.
\end{enumerate}
Then, counting the floating point operations needed to evaluate the objective function \eqref{eq:svd2} and the action of the linearization \eqref{eq:svd3}, the number of floating point operations needed for each optimization iteration is
\[
\mathcal{O}\big( (k_\tx{obj} + k_\tx{adj}) (m\,r + n\,r) \big)+ k_\tx{obj}\, C_1(n) + k_\tx{adj}\,C_2(n).
\]
The number of floating point operations needed to evaluate the determinant in \eqref{eq:det} is $\mathcal{O}(m\,r^2 + r^3) + r\,C_2(n)$. Thus, a total of 
\[
\mathcal{O}\big( k_\tx{opt}\,(k_\tx{obj} + k_\tx{adj}) (m\,r + n\,r) + m\,r^2 + r^3\big)+ k_\tx{opt}\,k_\tx{obj}\, C_1(n) + (k_\tx{opt}\,k_\tx{adj} + r) \,C_2(n),
\]
floating point operations are needed to compute one RTO sample, where the big--$\mathcal{O}$ term above refers to the total linear algebra cost, and the other terms refer to the total cost of evaluating $G(v)$ and the actions of $\nabla G(v)$.
\end{proposition}

Without loss of generality, the dimension $n$ of the parameters is often proportional to the number of degrees of freedom of the discretized forward model, and thus the functions $C_1(n)$ and $C_2(n)$ are often {\it linear} or {\it quasilinear} for scalable forward solvers, e.g., full multigrid solvers or preconditioned Krylov solvers.
In this case, the computational complexity of each optimization iteration is dictated by the cost of solving the forward model and evaluating actions with its linearization. Similarly, the computational complexity of evaluating the determinant in \eqref{eq:det} is dictated by the cost of the MVP with the linearized forward model. 
In contrast, without subspace acceleration, the complexity of computing the original objective function in \eqref{eq:optprob} is quadratic in $n$, the complexity of computing the action of the linearization $Q^\top \nabla H(v)$ is also quadratic in $n$, and the complexity of computing the determinant in the weighting function \eqref{eq:w(theta)} is cubic in $n$, since a dense matrix $Q \in \R^{(m+n)\times n}$ is involved. 
The cost of operating with the matrix $Q$ will thus dominate the overall computational cost of standard RTO for high-dimensional problems.
Subspace acceleration therefore significantly reduces the computational complexity of minimizing the RTO objective and calculating the determinant for each proposal. In addition, the size of the optimization problem in \eqref{eq:svd2} is also reduced to the intrinsic rank $r$. 

\paragraph{Rank truncation} For many inverse problems, the singular values of $\Jac G(v_\tx{ref})$ decay quickly, as a consequence of a smoothing forward operator, noisy observations, and the correlation structure of the prior (where some smoothness is necessary to make the Bayesian inverse problem well-posed \cite{StuartActa}). 
This fact is often used to reduce the parameter dimension of inverse problems by truncating the equivalent eigendecomposition \eqref{eq:eig} (cf.~\cite{bui,flath,spantini}), and hence to accelerate Markov chain Monte Carlo algorithms \cite{dili,lis,martin} and to approximate posterior distributions \cite{bui,flath,jointredu,zahm}.

Using intuition derived from optimal posterior approximations in linear Bayesian inverse problems \cite{spantini}, we can derive heuristics for truncating the SVD in scalable RTO. This can be particular useful for cases where data is abundant, i.e., when $y\in\R^m$ is a large vector. 
Suppose we have a linear inverse problem, that is, $G(v) = Gv$ and $H(v) = Hv$. Computing the reduced SVD $\nabla G(v) \equiv G = \Psi \Lambda \Phi^\T$, the inverse of the posterior covariance is given by the Gauss-Newton Hessian, which has the eigendecomposition\footnote{The linearization $\Jac H(v)$ does not depend on $v$ for linear inverse problems. We use this notation for consistency with the nonlinear case. }
\[ 
\Jac H(v_\tx{ref}) ^\T \Jac H(v_\tx{ref})  = \Phi (\Lambda^2 + \I_r) \Phi^\T + (\I_n - \Phi\Phi^\T).
\]
The subspace spanned by $\Phi$ contains the parameter directions where the posterior differs from the prior, since the prior (on the whitened variable $v$) has identity covariance matrix $\I_n$. 
A small singular value $\lambda_i$ implies that, along the corresponding right singular vector $\phi_i$, the variance reduction from prior to posterior is small; in particular, the ratio of posterior to prior variance is nearly one \cite{spantini}. We can thus neglect parameter directions corresponding to small singular values by truncating the SVD of $\Jac G(v_\tx{ref})$. Suppose that the truncation rank is $t < r$; this leads to an approximate eigendecomposition in the form of 
\begin{equation}
\Jac H(v_\tx{ref}) ^\T \Jac H(v_\tx{ref}) \approx \Phi_t (\Lambda_t^2 + \I_t)\Phi_t^\T + (\I_n - \Phi_t^{}\,\Phi_t^\T),\label{eq:eig_trunc}
\end{equation}
where $\Phi_t \in \R^{n \times t}$ and $\Lambda_t \in \R^{t \times t}$ consist of the leading $t$ right singular vectors and singular values, respectively.
In the linear case, the RTO proposal is a Gaussian distribution with the covariance matrix given by the inverse of the truncated approximation in \eqref{eq:eig_trunc}; this result directly follows from \eqref{eq:svdet1}. 
As shown in \cite{spantini}, the inverse of the approximation in \eqref{eq:eig_trunc} is also an optimal approximation to the posterior covariance with respect to the natural (geodesic) distance on the manifold of symmetric positive definite matrices. 
In this situation, truncating the SVD of $\Jac G(v_\tx{ref})$ for singular values that are smaller than one, e.g., $10^{-2}$ or $10^{-3}$, yields negligible impact on the RTO proposal. 

In nonlinear settings, we can adopt the same truncation strategy as a heuristic. The truncation will change RTO's map \eqref{eq:scale} and the resulting proposal distribution. 
Figure~\ref{fig:trunc_nonlin} shows the effect on RTO's proposal of truncating the SVD, in a toy example with a nonlinear forward model and a standard normal prior. 
Truncation restricts the role of the data misfit term in the construction of the proposal distribution. As the rank $r$ decreases, the proposal distribution becomes broader. In the extreme case, when $r$ is truncated to zero, RTO's proposal reverts to the prior. Note, however, the non-Gaussianity of the RTO proposal for $r \geq 1$ in this nonlinear example. We will evaluate the impact of truncation on MCMC sampling efficiency in subsequent numerical examples.

\begin{figure}[ht]
\centering
\begin{subfigure}{.32\textwidth}
\centering	
\pic{1}{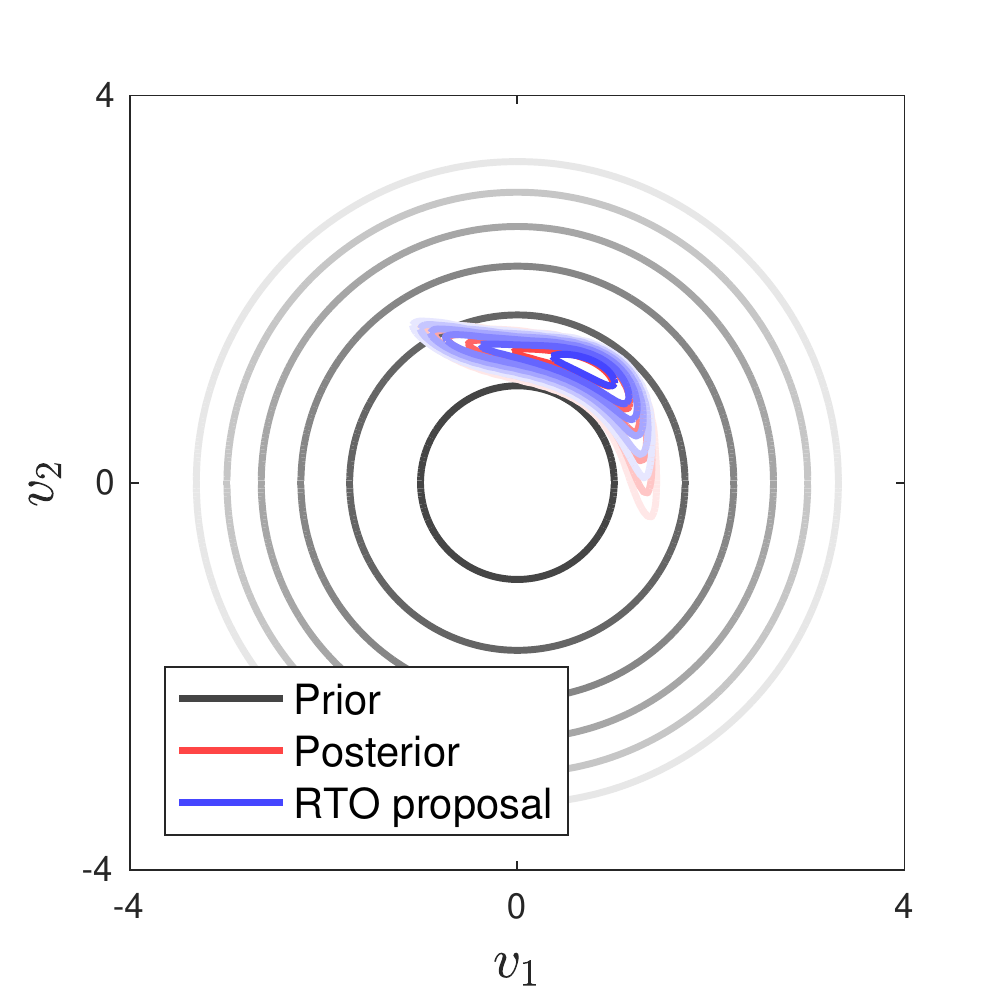}\\\pic{1}{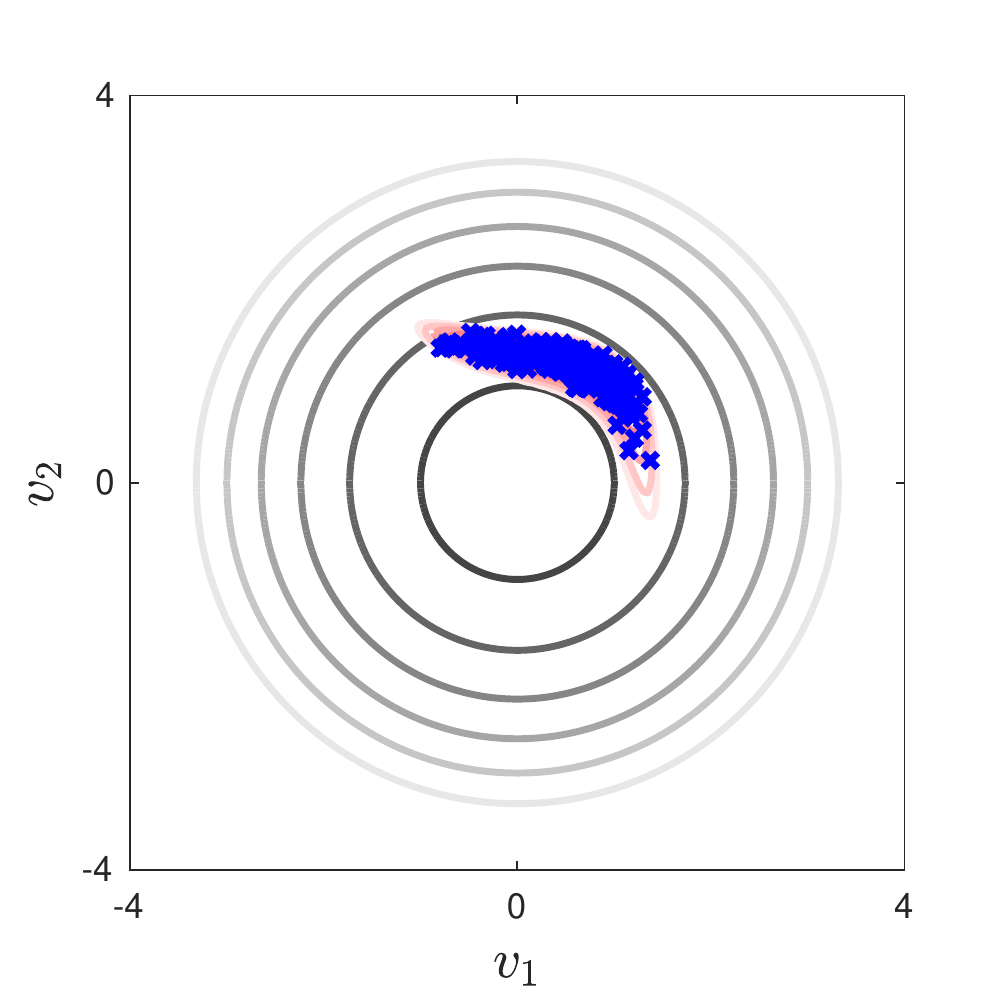}
\caption{Rank = $2$}
\end{subfigure}
\begin{subfigure}{.32\textwidth}
\centering
\pic{1}{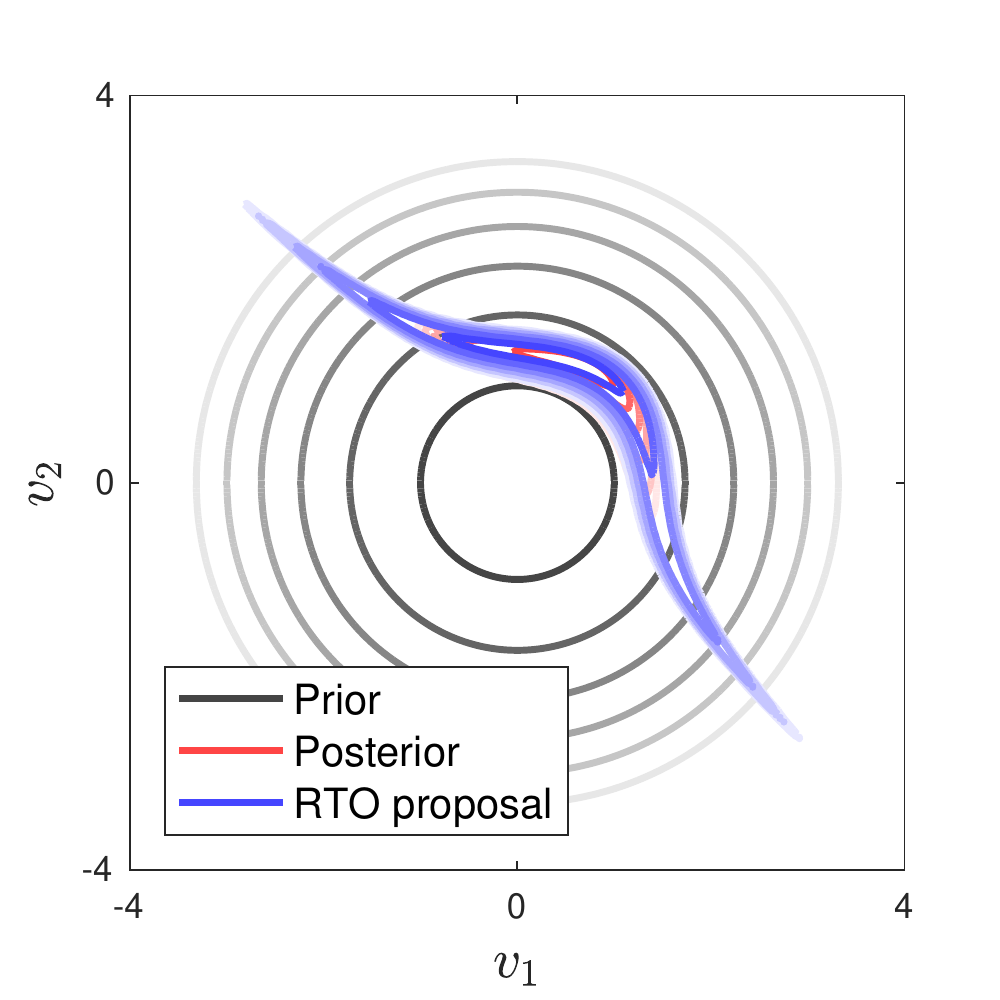}\\\pic{1}{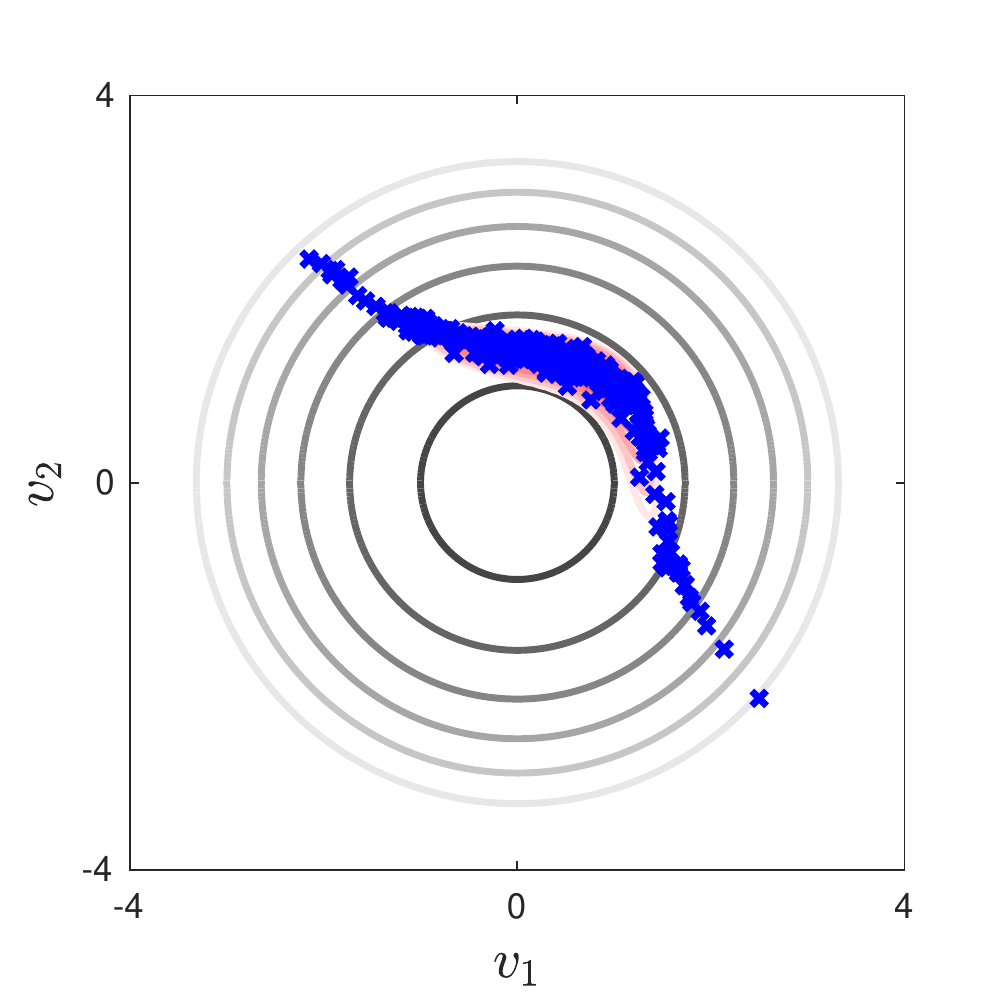}
\caption{Rank = $1$}
\end{subfigure}
\begin{subfigure}{.32\textwidth}
\centering
\pic{1}{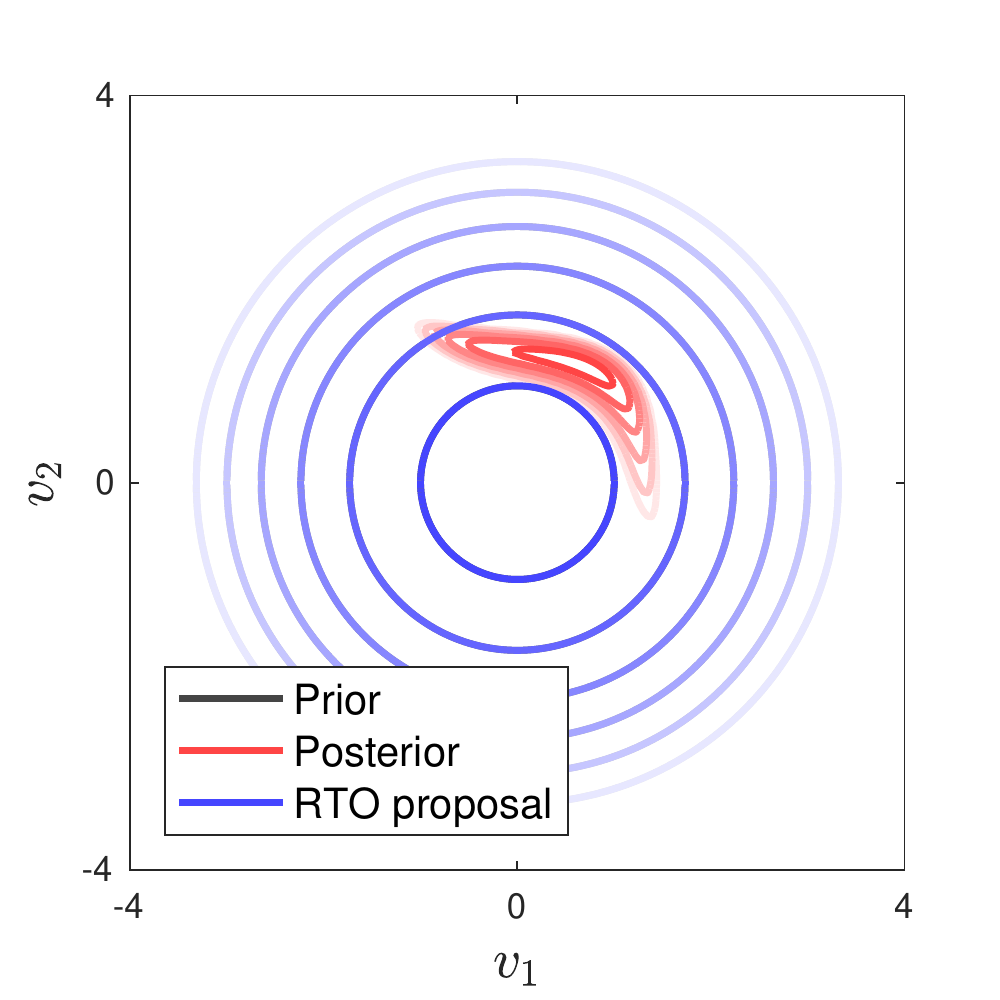}\\\pic{1}{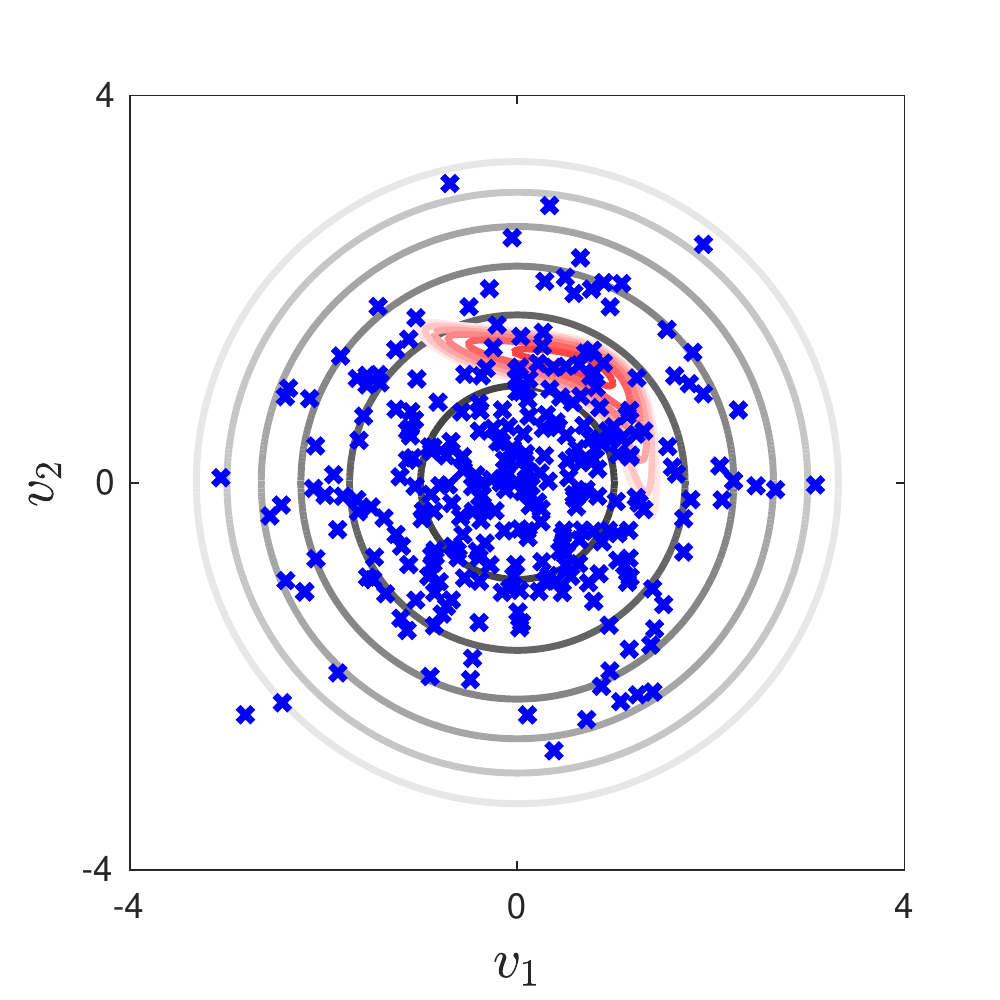}
\caption{Rank = $0$}
\end{subfigure}
\caption{Truncating the SVD in a two-dimensional toy example with a nonlinear forward model and standard normal prior. Top: contours of the prior, posterior and RTO's proposal density. Bottom: contours of the prior and posterior densities, and samples from RTO's proposal.}\label{fig:trunc_nonlin}
\end{figure}

\section{RTO on function space}\label{sec:proof}
The scalable implementation presented in Section \ref{sec:scalable} ensures that the computational cost of generating each RTO sample is dictated by the cost of evaluating the forward model and the adjoint model. 
When applying RTO as an independent proposal in the MH algorithm, or as the biasing distribution in self-normalised importance sampling, it is also critical to understand how its statistical performance (for instance, as measured by the acceptance rate of independence MH) depends on the dimension of the discretized parameters. 
To this end, we adopt the function space framework of \cite{StuartActa} to analyze the RTO proposal. We will focus on the case of applying RTO within MCMC, though the analysis can easily be adapted to importance sampling.
In this section, we will first provide background on MCMC in the function space setting, then interpret RTO's mapping in function space, and conclude by establishing sufficient conditions such that the statistical performance of RTO is invariant to the dimension of discretised parameters. 

\subsection{Function space MCMC}

To be aligned with the framework of \cite{StuartActa}, we will consider the target distribution on the original parameter $u$ (rather than the ``whitened'' parameter $v$), in a function space setting.
To preserve interpretability, we will use the same notation in the function space setting as we do in the finite dimensional setting to represent the parameters, prior mean, and prior covariance.
One exception is that we will use $\Gamma_\tx{pr}^{1/2}$ to denote the symmetric square root of the prior covariance operator, which is equivalent to any square root of the prior covariance up to a rotation.

We suppose that the parameter $u$ is an element of a separable Hilbert space $\cH$, endowed with a Gaussian prior measure $\mu_\tx{pr}$ such that the prior covariance $\Gamma_\tx{pr}$ is a self-adjoint, positive definite, and trace-class operator on $\cH$. 
The inner product on $\cH$ is denoted by $\langle \cdot \, , \cdot \rangle_\cH$, with the associated norm denoted by $\| \cdot \|_\cH$. For brevity, where misinterpretation is not possible, we will drop the subscript $\cH$. We assume that the data $y$ remain finite dimensional, i.e., $y \in \mathbb{R}^m$, $\Gamma_\tx{obs} \in \mathbb{R}^{m \times m}$, and $F: \cH \to \mathbb{R}^m$ for $m < \infty$. 
This way, the target probability measure is expressed by the Radon-Nikodym derivative
\[
\frac{\dd \mu_\tx{tar}}{\dd \mu_\tx{pr}}(u) \propto \exp\Big( - \frac12 ( y - F(u) )^\T \Gamma_\tx{obs}^{-1} ( y - F(u) ) \Big),
\]
with respect to the the prior measure.
The Metropolis--Hastings algorithm defines a Markov chain of random functions, asymptotically distributed according to the target measure, in the following way:
Given the current state of the Markov chain, $U^{(k)} = u$, a candidate state $u^{\prime}$ is drawn from a proposal $q(u,\cdot)$.
Define the following pair of measures on $\cH \times \cH$: 
\begin{equation}
\begin{array}{rll}
\label{eq:nus}
\nu(du,du^{\prime}) &=& q(u,du^{\prime}) \mu_\tx{tar}(du) \\
\nu^\bot(du,du^{\prime}) &=& q(u^{\prime},du) \mu_\tx{tar}(du^{\prime}).
\end{array}
\end{equation}
Then, the next state of the Markov chain is set to $U^{(k+1)} =
u^{\prime}$ with probability
\begin{equation} 
\alpha(u,u^{\prime}) = \min 
\Big \{1, 
\frac{d\nu^\bot}{d\nu}(u,u^{\prime})
 \Big \} , 
\label{acceptance} 
\end{equation}
and to $U^{(k+1)}=u$ otherwise.

For a continuously differentiable (as in Assumption~\ref{assumoverall}) and sufficiently bounded (as defined in Assumption 2.7 of \cite{StuartActa}) forward model $F$,
\cite{StuartActa} shows that the target measure is dominated by the prior measure.
As a result, refinements of the corresponding finite-dimensional target measure (induced by refinements of the parameter discretization) will converge to an infinite-dimensional limit. 
To make the acceptance probability of MH then invariant to parameter discretization, i.e., convergent to some positive infinite-dimensional limit and hence yielding a valid transition kernel \cite{MCMC:Tierney_1998}, we require the absolute continuity condition $\nu^\bot \ll \nu$.
We will refer to a MH algorithm as {\it well-defined} if this absolute continuity condition holds. 
Note that many Markov chain Monte Carlo methods designed for finite dimensional problems may not be well-defined on $\cH$---they may have vanishing acceptance probability and vanishing effective sample size with increasing parameter discretization dimension \cite{convrwm,mala}. 
For example, the acceptance probability of an MH algorithm using the standard random walk proposal scales as $\mathcal{O}(n^{-1})$ with parameter dimension, and thus it is not suitable for high-dimensional problems. 
We aim to show that MH with an RTO proposal is well-defined on $\cH$.

\subsection{RTO mapping in function space}
Recall that the Cameron--Martin space associated with the prior measure $\mu_\tx{pr}$, $\cHcm = \Gamma_\tx{pr}^{1/2}\,\cH  \subset \cH$, is equipped with the inner product
\[
\langle a, b \rangle_{\cHcm} \defeq \langle a, b \rangle_{\Gamma_\tx{pr}^{-1}} = \lb\langle \Gamma_\tx{pr}^{-\frac12}\, a, \Gamma_\tx{pr}^{-\frac12}\,b \rb\rangle, \]
for any $a, b \in \cHcm$.
Here we will show that a sample generated by the RTO mapping is a modification of a random function drawn from the prior measure along a finite dimensional subspace of the Cameron--Martin space.

We first consider the properties of the rank-$r$ reduced SVD of the linearized forward model $\Jac G(v_\tx{ref}) = \Psi \Lambda \Phi^\T$ in the function space setting. 
The right singular vectors $\phi_1, \phi_2, \ldots, \phi_r$ are also eigenfunctions of the eigenvalue problem
\[
\Jac G(v_\tx{ref})^\natural \Jac G(v_\tx{ref}) \, \phi_i = \lambda_i \phi_i, \quad i=1,\ldots, r,
\]
where $\Jac G(v_\tx{ref})^\natural$ denotes the adjoint of the operator $\Jac G(v_\tx{ref})$.
Recalling the whitening transform introduced in Section \ref{sec:target}, we have 
\[
\Jac G(v_\tx{ref})^\natural \Jac G(v_\tx{ref}) = \Gamma_\tx{pr}^{\frac12} \Jac F\lb(u_\tx{ref}\rb)^\natural \Gamma_\tx{obs}^{-1} \Jac F\lb(u_\tx{ref}\rb) \Gamma_\tx{pr}^{\frac12},
\]
where $ \Jac F : \cH \rightarrow \R^m$ is the Fr\'{e}chet derivative of the forward model and $\Jac F\lb(u_\tx{ref}\rb)^\natural$ is its adjoint. 
Defining a new set of functions 
\begin{equation}
\chi_i = \Gamma_\tx{pr}^{1/2} \phi_i,
\end{equation} 
we also have an equivalent eigenvalue problem
\[
\Gamma_\tx{pr} \Jac F\lb(u_\tx{ref}\rb)^\natural \Gamma_\tx{obs}^{-1} \Jac F\lb(u_\tx{ref}\rb) \, \chi_i = \lambda^2_i \chi_i, \quad i=1,\ldots, r.
\]
Since the operator $\Jac F\lb(u_\tx{ref}\rb)^\natural \Gamma_\tx{obs}^{-1} \Jac F\lb(u_\tx{ref}\rb)$ is self-adjoint and has finite rank $r \leq m$ in the case of finite-dimensional data ($m < \infty$), we have that the eigenfunctions $\chi_i \in \Gamma_\tx{pr} \cH$ and that the right singular vectors  $\phi_i \in \Gamma_\tx{pr}^{1/2} \cH$, for $i=1,\ldots, r$.
\begin{remark}
Both $\{\chi_1, \chi_2, \ldots, \chi_r \}$ and $\{\phi_1, \phi_2, \ldots, \phi_r \}$ span finite dimensional subspaces
in the \\Cameron--Martin space.
The basis functions $\{\phi_1, \phi_2, \ldots, \phi_r \}$ are orthogonal with respect to the inner product $\langle\cdot, \cdot\rangle$, whereas the basis functions $\{\chi_1, \chi_2, \ldots, \chi_r \}$ are orthogonal with respect to the Cameron--Martin inner product $\langle\cdot, \cdot\rangle_{\Gamma_\tx{pr}^{-1}}$.
\end{remark}

In Figure~\ref{fig:quadrv}, we specify the relationship between four random variables: $u, \zeta \in \cH$, and $v, \xi\in  \Gamma_\tx{pr}^{-\frac12}\cH$, where $\zeta$ is a newly defined random variable distributed according to the prior. 
Here we use $u$ and $v$ to denote random variables distributed according to the unwhitened and whitened RTO measures, respectively, rather than the corresponding target measures. 
This way, we have the identity
\[
\langle v, \phi_i \rangle = \langle \Gamma_\tx{pr}^{-\frac12} (u - m_\tx{pr}) , \Gamma_\tx{pr}^{-\frac12} \chi_i \rangle = \langle u - m_\tx{pr} , \chi_i \rangle_{\Gamma_\tx{pr}^{-1}},
\]
for any $v \in \Gamma_\tx{pr}^{-\frac12}\cH$ and any right singular vector $\phi_i \in \Gamma_\tx{pr}^{1/2}\cH$.

\begin{figure}[h]
\centering
\pic{1}{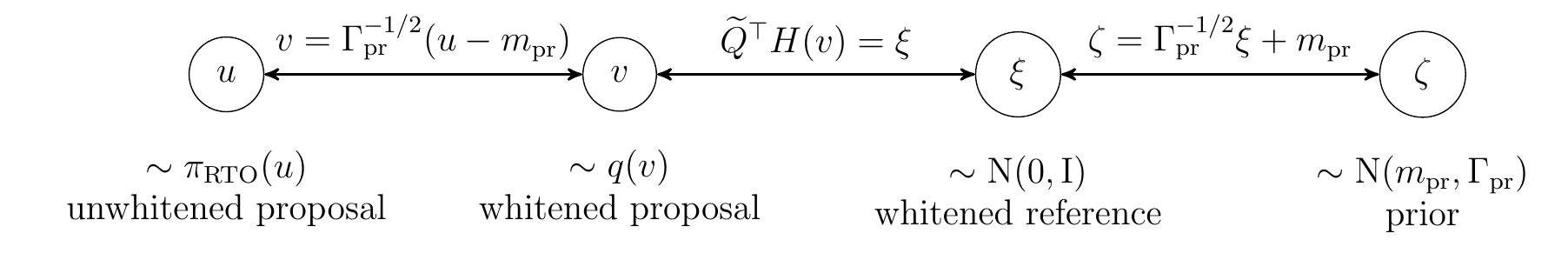}
\caption{Relationship between four random variables, $u, \zeta \in \cH$, $v, \xi\in  \Gamma_\tx{pr}^{-\frac12}\cH$.}\label{fig:quadrv}
\end{figure}%

Given the basis functions $\{\chi_1, \chi_2, \ldots, \chi_r \}$, we introduce a linear map $R : \cH \rightarrow \R^r$ whose components are
\[
R_i(u) = \langle u , \chi_i \rangle_{\Gamma_\tx{pr}^{-1}}, \quad i = 1, \ldots, r,
\]
and a projector $P: \cH \rightarrow \tx{span}\{\chi_1, \chi_2, \ldots, \chi_r \}$ specified as
\[
P u = \sum_{i = 1}^{r} \chi_i \, R_i(u).
\]
For a random function $\zeta \in \cH$ drawn from the prior measure, we can then express the RTO mapping \eqref{eq:scale} in the unwhitened coordinates:
\begin{equation} 
\label{eq:oblique} 
\lb\{\begin{array}{r l} (\I-P) (\zeta-m_\tx{pr})  & = (\I-P) (u-m_\tx{pr}) \\ 
P \;(\zeta- m_\tx{pr}) & = X \Big[ (\Lambda^2 + \I)^{-\frac12} \big( R(u-m_\tx{pr}) + \Lambda  \Psi^\T S_\tx{obs}^{-1} (F(u)-y) \big)\Big],
\end{array}\rb. 
\end{equation}
where $X = [\chi_1, \chi_2, \ldots, \chi_r ]$. 
Analogous to the splitting of the RTO solution in the scalable implementation, we define
\begin{equation} \label{eq:inf1} 
u_r = R(u - m_\tx{pr}), \quad \tx{and} \quad u = X u_r + u_\perp + m_\tx{pr},
\end{equation}
and projected random variables
\begin{equation} \label{eq:inf2} 
\zeta_r = R(\zeta - m_\tx{pr}), \quad \tx{and} \quad \zeta_\perp = (\I-P) (\zeta-m_\tx{pr}),
\end{equation} 
where $u_\perp$ and $\zeta_\perp$ are elements of the complement of $\tx{range}(X)$.
Then, we can solve the nonlinear system of equations \eqref{eq:oblique} by first letting $u_\perp = \zeta_\perp$ and then solving the $r$--dimensional system of equations
\begin{equation}\label{eq:inf3} 
\Theta(u_r; u_\perp) = \zeta_r,
\end{equation}
for a given $u_\perp$, where the function 
\begin{equation} \label{eq:inf4} 
\Theta(u_r; u_\perp) = (\Lambda^2 + \I)^{-\frac12}  \Big[u_r + \Lambda \Psi^\T S_\tx{obs}^{-1} (F(X u_r + u_\perp + m_\tx{pr})-y) \Big],
\end{equation}
is an affine transformation of the nonlinear forward model $F$.

\begin{remark}
For problems with finite dimensional data, the RTO mapping necessarily modifies a random function drawn from the prior measure only in the finite dimensional subspace spanned by $\{\chi_1, \chi_2, \ldots, \chi_r\}$, which is a subspace of the Cameron--Martin space. 
This is well aligned with the nature of function space inverse problems, where the update from the prior to the posterior is expected to take place in the Cameron--Martin space. Other accelerations of function space MCMC, e.g., \cite{dili,MCMC:DS_2018}, adopt similar approaches to modify their algorithms in some finite dimensional subspace of the Cameron--Martin space. 
For problems with functional (infinite-dimensional) data, as long as the equivalence between the target measure and the prior measure can be established,
one can truncate the SVD and apply RTO as in the finite data case. Such a truncation can also be the key to managing overall computational complexity. 
\end{remark}

\subsection{Well-definedness of RTO on function space}\label{sec:inf_rnd}
We will first prove that, under certain conditions, the RTO measure $\pi_\tx{RTO}$ is equivalent to the prior $\mu_\tx{pr}$. 
Under these conditions, we can then show that RTO is well defined on $\cH$.

\begin{theorem}\label{theo:rtoeqpr}
Let $\zeta$ be a random variable distributed according to the prior measure $\mu_\tx{pr}$, $u$ be the random variable defined through the mapping in (\ref{eq:oblique}), and $\mu_\tx{RTO}$ be the measure induced by $u$.
Denote the subspace $\tx{span(\chi_1, \chi_2, \ldots, \chi_r)}$ and its complement by $W$ and $W^\perp$, respectively. Without loss of generality, let the prior mean be zero. Suppose that for all $a\in \R^r$ and $b \in W^\perp$, the mapping
\[ 
a \mapsto a + \Lambda \Psi^\T S_\tx{obs}^{-1} F(X a + b) ,
\]
is Lipschitz continuous, injective, and its inverse is Lipschitz continuous.  Then, the RTO measure $\mu_\tx{RTO}$ is equivalent to the prior $\mu_\tx{pr}$.
\end{theorem}

\begin{proof}
In this proof only, we will employ the probability triplet $(\Omega,\cF,\uP)$ and describe the random function formally as the map $u : \Omega \to \cH$. 
We assume that the measurable space $\lb(\Omega, \cF\rb)$ is a Radon space. The four random variables in \eqref{eq:inf1}--\eqref{eq:inf3} can be defined as
\begin{align*}
\zeta_r&: \Omega \to \R^r, &	u_r&: \Omega \to \R^r, & \zeta_\perp&: \Omega \to W^\perp, & u_\perp&: \Omega \to W^\perp  .
\end{align*}
We use $\mathcal{B}(\cdot)$ to denote the Borel algebra. Let the notation $\uP^u$ denote the push-forward measure of $\uP$ through the mapping $u$
\[ 
\uP^u \defeq u_\sharp \uP = \uP \lb(u^{-1}(\cdot)\rb) = \uP\lb( u \in\, \cdot\,\rb).
\]
Using this notation, we have $\mu_\tx{pr} = \uP^\zeta$ and $\mu_\tx{RTO} = \uP^u$.

Since, under the mapping \eqref{eq:oblique}, the infinite dimensional random variables $\zeta_\perp$ and $u_\perp$ take the same value, we use the regular conditional probability $\nu:W^\perp \times \cF \to [0,1]$ of the form
\[
\nu : (b,A) \to \nu(b,A) = \uP\lb(u\in A \,\mb\vert\, u_\perp = b\rb), \quad \forall A \in \cF, \; \forall b \in W^\perp,
\]
to analyze the RTO measure. For any $A \in \cB(\cH)$ and $b \in W^\perp$, we define the set
\[
A_r(b) \defeq \{ a \in \R^r \,|\, Xa + b \in A\}.
\]
Then, for any $A \in \cB(\cH)$, the RTO measure can be expressed in a conditional form
\[
\uP^u(A) = \uP\lb( \{u \in A\} \intersect \{u_\perp \in W^\perp\} \rb) = \int_{W^\perp} \uP\lb(u \in A \,|\, u_\perp = b\rb) \uP^{u_\perp} (\dd b).
\]
Given a fixed $b \in W^\perp$, RTO solves the equation $\Theta(u_r; b) = \zeta_r$, and thus we have
\[
\uP\lb( u \in A \,|\, u_\perp = b \rb)  = \uP\lb( u_r \in A_r(b) \,|\, u_\perp = b \rb) = \uP\lb(\zeta_r \in \Theta(A_r(b); b) \,|\, \zeta_\perp = b \rb) = \nu\lb(b, \zeta_r^{-1}\circ \Theta(A_r(b); b)\rb).
\]
This way, the RTO measure can be expressed in terms of the measure $\uP^{u_\perp} = \uP^{\zeta_\perp}$ and the conditional measure $\nu\lb(b, \zeta_r^{-1} (\cdot)\rb)$. This leads to 
\begin{equation}\label{eq:rto_law}
\uP^u(A) = \int_{W^\perp} \nu\lb(b, \zeta_r^{-1}\circ \Theta(A_r(b); b)\rb) \uP^{\zeta_\perp} (\dd b).
\end{equation}

The conditional measure $ \nu\lb(b, \zeta_r^{-1}(\cdot)\rb) = \uP(\zeta_r \in \cdot \, |\, \zeta_\perp = b)$ is the measure of a finite number of directions, $\zeta_r$, of the prior conditioned on a particular value, $\zeta_\perp = b$. 
Since the basis functions $\{\chi_1, \chi_2, \ldots, \chi_r \}$ are orthogonal with respect to the Cameron--Martin inner product $\langle\cdot, \cdot\rangle_{\Gamma_\tx{pr}^{-1}}$, the two random variables $\zeta_r$ and $\zeta_\perp$ are independent (see Proposition 1.26 of \cite{InfDimAnalysis}). Hence, the conditional measure $\nu\lb(b, \zeta_r^{-1}(\cdot)\rb)$ is equivalent to the law of $\zeta_r$. 
Let $\pi_{\zeta_r}$ denote its probability density function. Then, we have
\begin{align}\label{eq:cond_law}
\nu\lb(b, \zeta_r^{-1}\circ \Theta(A_r(b); b)\rb) &= \int_{\Theta(A_r(b);b)} \pi_{\zeta_r} (a) \dd a \nonumber \\
&= \int_{A_r(b)} \pi_{\zeta_r}\circ \Theta(a;b)  \lb| \det \Jac \Theta(a;b) \rb| \dd a \nonumber\\
&= \int_{A_r(b)} \frac{\pi_{\zeta_r}\circ \Theta (a;b)}{\pi_{\zeta_r} (a)} \lb| \det \Jac \Theta(a;b) \rb| \, \pi_{\zeta_r} (a) \;\dd a \nonumber\\
&= \int_{A_r(b)} \frac{\pi_{\zeta_r}\circ \Theta (a;b)}{\pi_{\zeta_r} (a)} \lb| \det \Jac \Theta(a;b) \rb| \, \nu\lb(b, \zeta_r^{-1}(\dd a)\rb).
\end{align}
The change of variables in the above expression uses the fact that $\Theta(\,\cdot\,; b)$ is Lipschitz continuous and injective, and that its inverse is Lipschitz continuous. 
Note that for almost all $b \in W^\perp$ and $a \in \R^r$, the expression
\[
\cR (a;b) \defeq \frac{\pi_{\zeta_r}\circ \Theta (a;b)}{\pi_{\zeta_r} (a)} \lb| \det \Jac \Theta(a;b) \rb| ,
\]
is positive. 
Substituting \eqref{eq:cond_law} into \eqref{eq:rto_law} and using the change of variables $b = (\I-P)z$ and $a = R(z)$ for any $z \in \cH$, we obtain the RTO measure in the form
\begin{equation}
\uP^u(A) = \int_{W^\perp} \int_{A_r(b)} \cR (a;b) \, \nu\lb(b, \zeta_r^{-1}(\dd a)\rb) \uP^{\zeta_\perp} (\dd b)  =  \int_{A}  \cR \Big(R(z); (\I - P)z \Big)  \uP^{\zeta} (\dd z).
\end{equation}
Therefore, the Radon--Nikodym derivative of the RTO measure with respect to the prior measure 
\[\frac{\dd \mu_\tx{RTO}}{\dd \mu_\tx{pr}}(u) = \cR \Big(R(z); (\I - P)z \Big), \]
is positive almost everywhere. This implies that $\mu_\tx{RTO}$ is equivalent to $\mu_\tx{pr}$. 
\end{proof}

Theorem \ref{theo:rtoeqpr} implies that the MH algorithm using RTO as its independence proposal yields dimension-independent performance in the function space setting. We formalize this notion in the following theorem.

\begin{theorem}
Suppose that the target measure $\mu_\tx{tar}$ is equivalent to the prior measure $\mu_\tx{pr}$, i.e., $\mu_\tx{tar} \sim \mu_\tx{pr}$. Under the assumptions of Theorem \ref{theo:rtoeqpr}, the acceptance probability of the MH algorithm using RTO as its independence proposal is positive almost surely with respect to $\mu_\tx{pr} \times \mu_\tx{pr}$. %
\end{theorem}
\begin{proof}
Following the result of Theorem \ref{theo:rtoeqpr} and the condition $\mu_\tx{pr} \sim \mu_\tx{tar}$, the Radon--Nikodym derivative of the target measure with respect to the RTO measure 
\[
\omega(u) = \frac{\dd \mu_\tx{tar}}{\dd \mu_\tx{RTO}}(u), 
\]
is $\mu_\tx{pr}$--almost surely positive.
The rest of the proof is a special case of Theorem 5.1 in \cite{StuartActa}. 
Since RTO is an independent proposal, the resulting MH proposal measure becomes $q(u,\dd u') = \mu_\tx{RTO}(\dd u')$, and the pair of transition measures of MH become
\begin{align*}
	\nu(\dd  u,\dd u') & = \mu_\tx{RTO}(\dd u')\, \mu_\tx{tar}(\dd  u)\\
	\nu^\perp(\dd  u,\dd u') & = \mu_\tx{RTO}(\dd u) \, \mu_\tx{tar}(\dd  u'),
\end{align*}
This way, the acceptance probability can be expressed as
\[
\alpha(u, u')  = \min \left( 1, \frac{\dd \mu_\tx{tar}}{\dd \mu_\tx{RTO}}(u') \Big/ \frac{\dd \mu_\tx{tar}}{\dd \mu_\tx{RTO}}(u) \right) = \min \left( 1, \frac{\omega(u')}{\omega(u)} \right).
\]
Because $\omega(u)$ is positive $\mu_\tx{pr}$--almost surely, the acceptance probability $\alpha(u, u')$ is positive $\mu_\tx{pr}\times \mu_\tx{pr}$--almost surely.
\end{proof}

RTO--MH is therefore well-defined in a function-space setting, under the conditions in Theorem \ref{theo:rtoeqpr}. Thus refining the parameter discretization in a discrete setting should not diminish RTO--MH's sampling efficiency. %
Note that the Radon--Nikodym derivative $\omega(u)$ is also the importance ratio used in self-normalized importance sampling. Ensuring that $\omega(u)$ is positive (almost surely) can make the effective sample size of the self-normalized importance sampling estimator invariant to the discretized parameter dimension; see \cite{agapiou} and references therein for formal justifications. 
\section{Example 1: 1D elliptic PDE} \label{sec:numerics}
The previous section provided a theoretical argument for RTO's dimension independence. This section numerically explores the factors that influence its sampling performance, using a simple one-dimensional elliptic PDE inverse problem. 
We describe the setup of the test case (Section~\ref{sec:probsetup}) and then explore the effects of parameter dimension (Section~\ref{sec:rto_dim}) and observational noise (Section~\ref{sec:rto_sig}). We conclude by comparing the performances of RTO and pCN (Section~\ref{sec:rto_pcn}).

\subsection{Problem setup}\label{sec:probsetup}
The diffusion equation is used to model the spatial distribution of many physical quantities, such as temperature, electrostatic potential, or pressure in porous media. We consider the following stationary diffusion equation, 
\begin{align*}
-\frac{\dd }{\dd x} \lb(\kappa(x) \frac{\dd p}{\dd x}(x) \rb)= f(x),\quad  0< x< 1, 
\end{align*}
with boundary conditions
\begin{align*}
\kappa(0) \frac{\dd p}{\dd x}(0)  = -1,\quad\quad p(1) = 1,
\end{align*}
and source term $f$. The diffusion coefficient $\kappa$ is endowed with a log-normal prior distribution. In particular, $\log \kappa$ is a Gaussian process with a Laplace-like differential operator as its precision operator. After discretization on a uniform grid with $n$ nodes, $\kappa$ is thus specified as
\begin{align*} \kappa &= 1.5 \exp\lb( S^{}_\tx{pr} v\rb) + 0.1 & S^{-1}_\tx{pr} &= \sqrt{n}\begin{bmatrix}
\sqrt{n} &  &  &  & \sqrt{n} \\ 
-1 & 1  &  &  & \\ 
 & -1 & 1  &  & \\ 
 &  &  &\ddots  & \\ 
 &  &  & -1 & 1 
\end{bmatrix} \end{align*}
where $v \in \mathbb{R}^n$ is a vector of independent standard normals and we have abused notation so that $\kappa \in \mathbb{R}^n$ immediately above as well.
For any realization of $\kappa$, the equation is solved numerically using finite differences with the three point central difference stencil. Derivatives of the potential field $p$ with respect to $\kappa$ are evaluated using the matrix-free adjoint model. %
In this setting, the dimension of discretized parameters is the same as the degrees of freedom in the forward model. Computing $ S^{}_\tx{pr} v$, solving the forward model, and solving the adjoint model (for one matrix-vector product with the Jacobian) all require $\mathcal{O}(n)$ floating point operations.

For the inverse problem, we suppose that the potential field is observed, with additive Gaussian noise, at nine equally-spaced points along the domain. Our goal is to condition the field $\kappa$ on these observations. We generate synthetic data using a mesh size of $151$, which does not correspond to any mesh size used in solving the inverse problem, avoiding an inverse crime. The ``true'' diffusion coefficient, source term, potential field, and data are depicted in Figure~\ref{fig:forcing}. 

\begin{figure}[ht]
\centering
\begin{subfigure}{0.4\textwidth}
\centering
\pic{1}{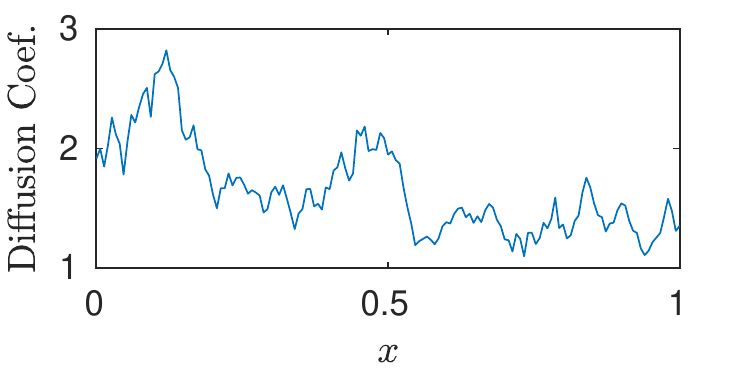}
\caption{``True'' diffusion coefficient}
\end{subfigure}
\begin{subfigure}{0.4\textwidth}
\centering
\pic{1}{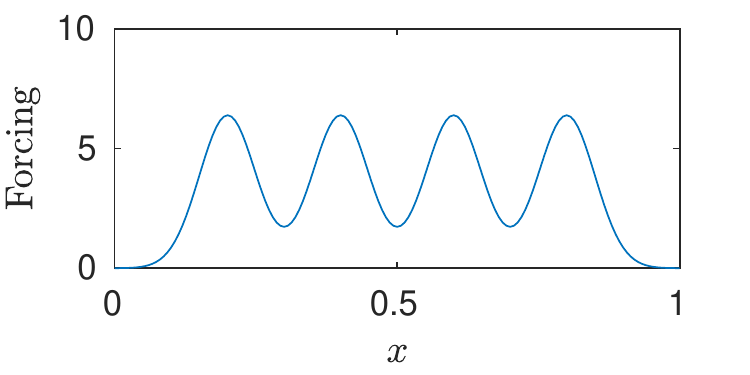}
\caption{Source term}
\end{subfigure}
\begin{subfigure}{0.4\textwidth}
\centering
\pic{1}{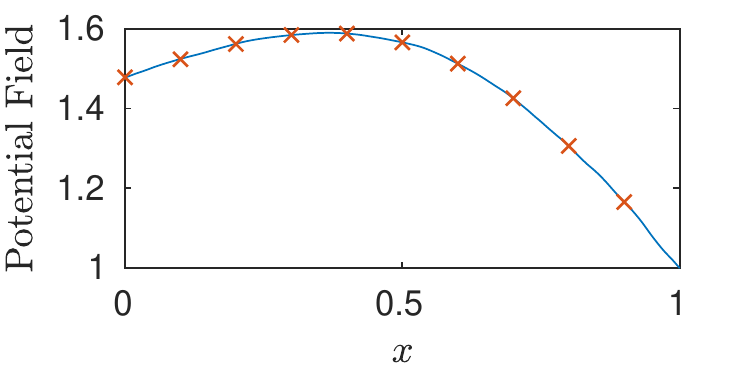}
\caption{Potential field (line) and observational data (x's)}
\end{subfigure}
\caption{Elliptic PDE problem setup.}\label{fig:forcing}
\end{figure}

\subsection{Influence of parameter dimension} \label{sec:rto_dim}
In our first experiment, we solve the Bayesian inverse problem using RTO for a series of parameter dimensions ranging from $n=41$ to $n=10241$. We fix the observational noise standard deviation %
to $10^{-5}$ and, at each parameter dimension, run an MCMC chain of $5000$ steps. The chains are started at the posterior mode.
As shown in Figure~\ref{fig:rto_post_dim}, the posterior distributions obtained for the different discretizations match quite closely. 
As shown in Table~\ref{tab:rto_dim}, the acceptance rate and effective sample size (ESS) are both high and essentially constant with respect to parameter dimension. (We report the median ESS over all components of the $n$-dimensional chain.)
These results provide an empirical demonstration of RTO's dimension independence, meaning that the number of MCMC steps required to obtain a single effectively independent sample is independent of $n$. 

The number of optimization iterations in each MCMC step is also roughly constant in $n$. 
To solve each optimization problem, we use the nonlinear least-squares solver in MATLAB, provided with Jacobian-vector products. The solver uses a trust-region-reflective algorithm where each iteration approximately solves a large linear system using preconditioned conjugate gradients. We set the starting point for each sequence of optimization iterations to the posterior mode. The primary stopping criterion is a function tolerance (i.e., a lower bound on the change in the value of the objective) of $10^{-6}$, which is below the level of discretization error. %

\begin{figure}[ht]
\centering
\begin{subfigure}{0.49\textwidth}
\centering
\pic{1}{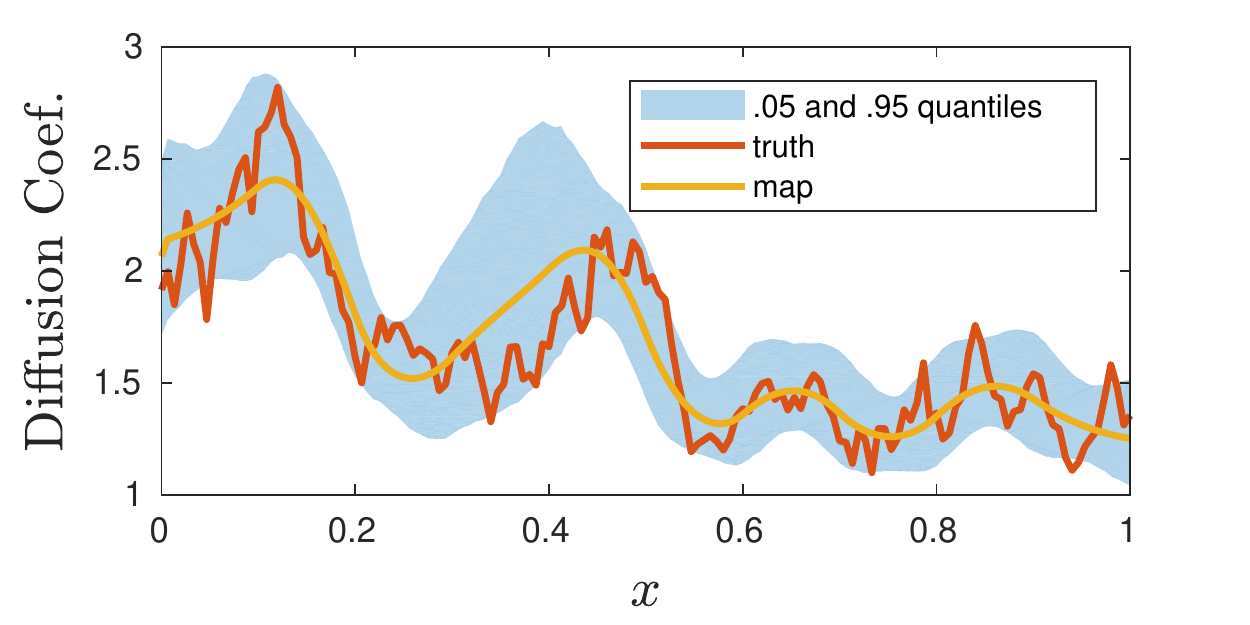}
\caption{$n = 161$}
\end{subfigure}
\begin{subfigure}{0.49\textwidth}
\centering
\pic{1}{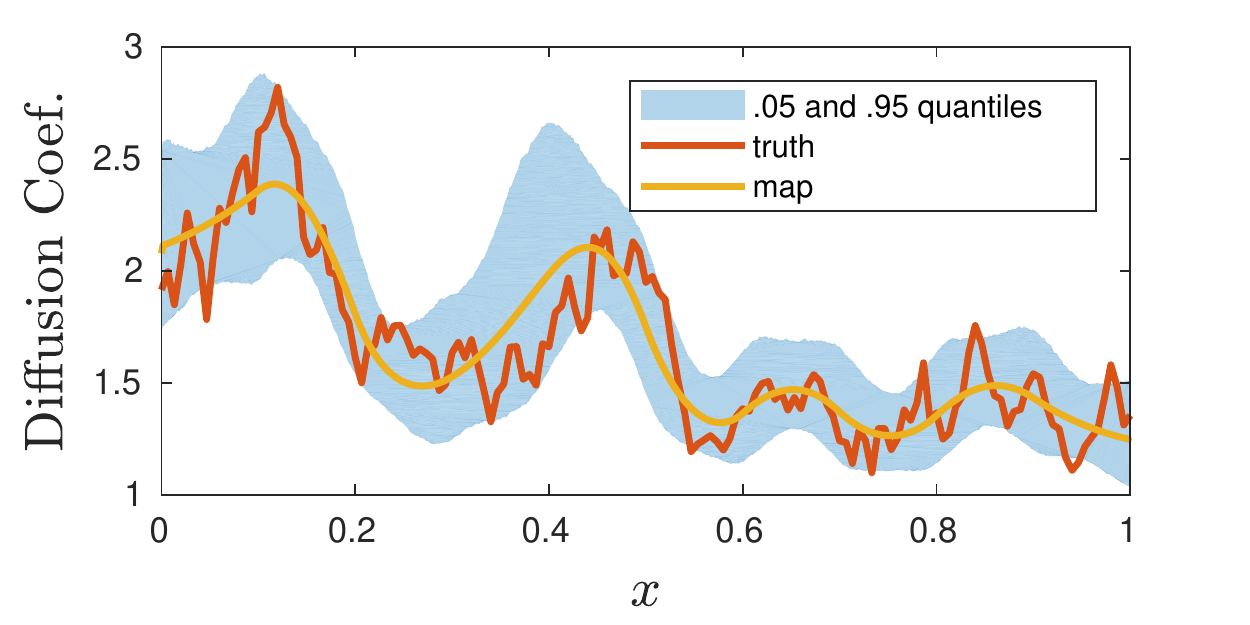}
\caption{$n = 641$}
\end{subfigure}
\begin{subfigure}{0.49\textwidth}
\centering
\pic{1}{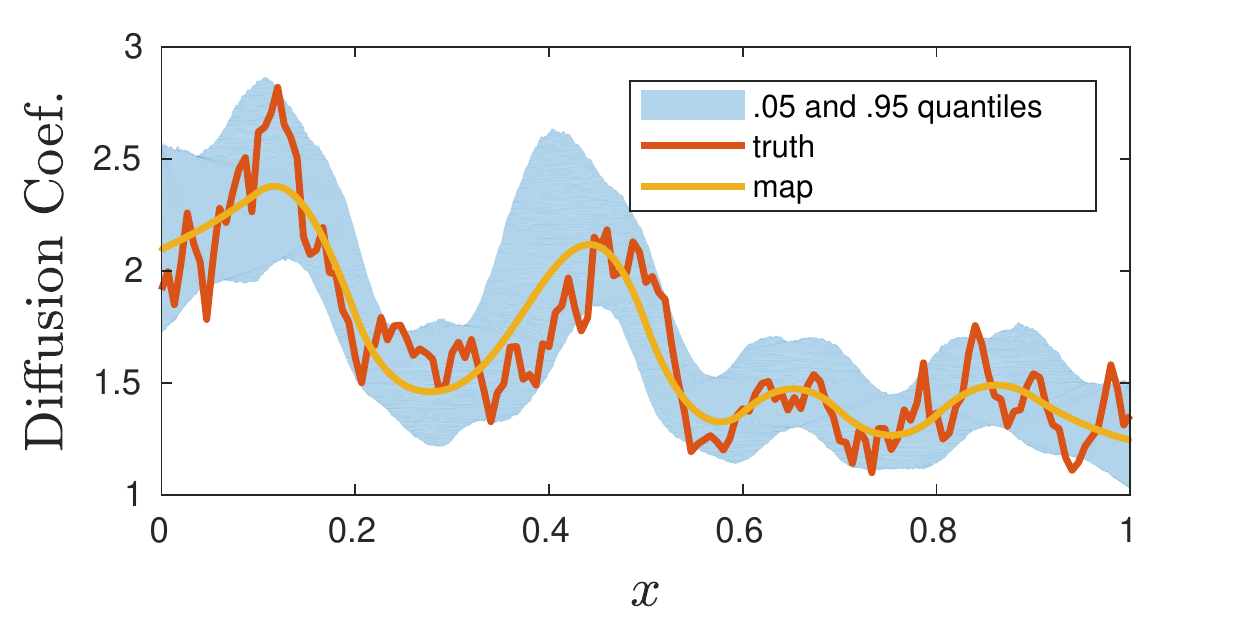}
\caption{$n = 2561$}
\end{subfigure}
\begin{subfigure}{0.49\textwidth}
\centering
\pic{1}{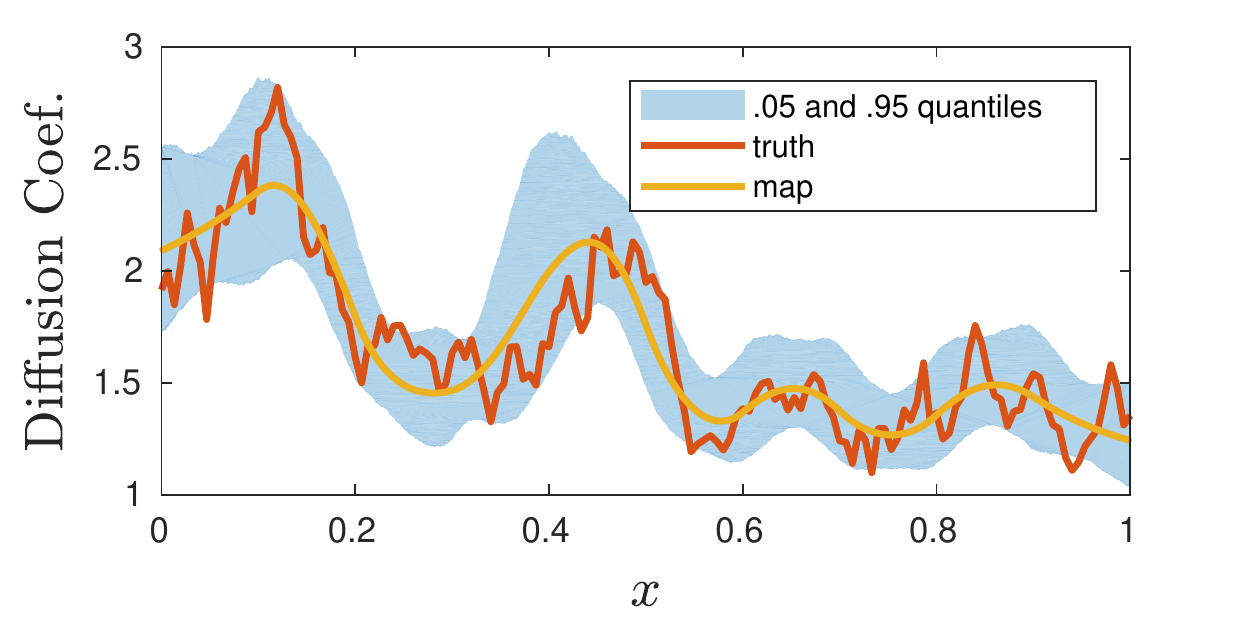}
\caption{$n = 10241$}
\end{subfigure}
\caption{Summary statistics of posterior distributions computed via RTO-MH with varying parameter dimension $n$. 90\% marginal credibility intervals (blue shaded region), true diffusivity coefficient (red line), and MAP estimate (yellow line).}\label{fig:rto_post_dim}
\end{figure}

\begin{table}[htbp]
{\centering
\caption{Effective sample size (ESS), average acceptance rate, and average number of optimization iterations per step of RTO, with varying parameter dimension. MCMC chain length is 5000 steps.}\label{tab:rto_dim}
\small
\begin{tabular}{l r r r r r r r r r}
\toprule
Parameter Dim. & $41$ & $81$ & $161$ & $321$ & $641$ & $1281$ & $2561$ & $5121$ & $10241$ \\
\midrule
ESS & 4268.9 & 4206.7 & 4307.1 & 4343.5 & 4544.8 & 4464.5 & 4523.3 & 4484.9 & 4532.2 \\
Acceptance Rate & 0.928 & 0.926 & 0.932 & 0.936 & 0.948 & 0.950 & 0.954 & 0.950 & 0.953 \\
Opt.\ Iterations & 170.74 & 209.12 & 273.03 & 324.04 & 357.76 & 307.50 & 198.81 & 165.06 & 142.25 \\
\bottomrule
\end{tabular}\\
}%
\end{table}

Figure~\ref{fig:cpu_vs_dim} shows the CPU times needed to generate one effectively independent sample, to take generate one RTO sample, and to evaluate the forward model once. All three lines suggest that the CPU times increases linearly with the discretized parameter dimension. This confirms our analysis of the computational complexity of RTO in Proposition \ref{prop:complexity}---in this example, the computational complexities of both the forward model and the RTO map are linear. It also implies that it takes the same number of MCMC steps to obtain a desired accuracy regardless of the discretized parameter dimension. This confirms our finding in Section \ref{sec:proof}. 
We also report the CPU time for the standard RTO (with a dense matrix $Q \in \R^{(m+n)\times n}$) to generate one sample (red line in Figure~\ref{fig:cpu_vs_dim}). In this example, we observe that the computational complexity of the standard RTO is quadratic with the parameter dimension. %

\begin{figure}[ht]
\centering
\pic{0.7}{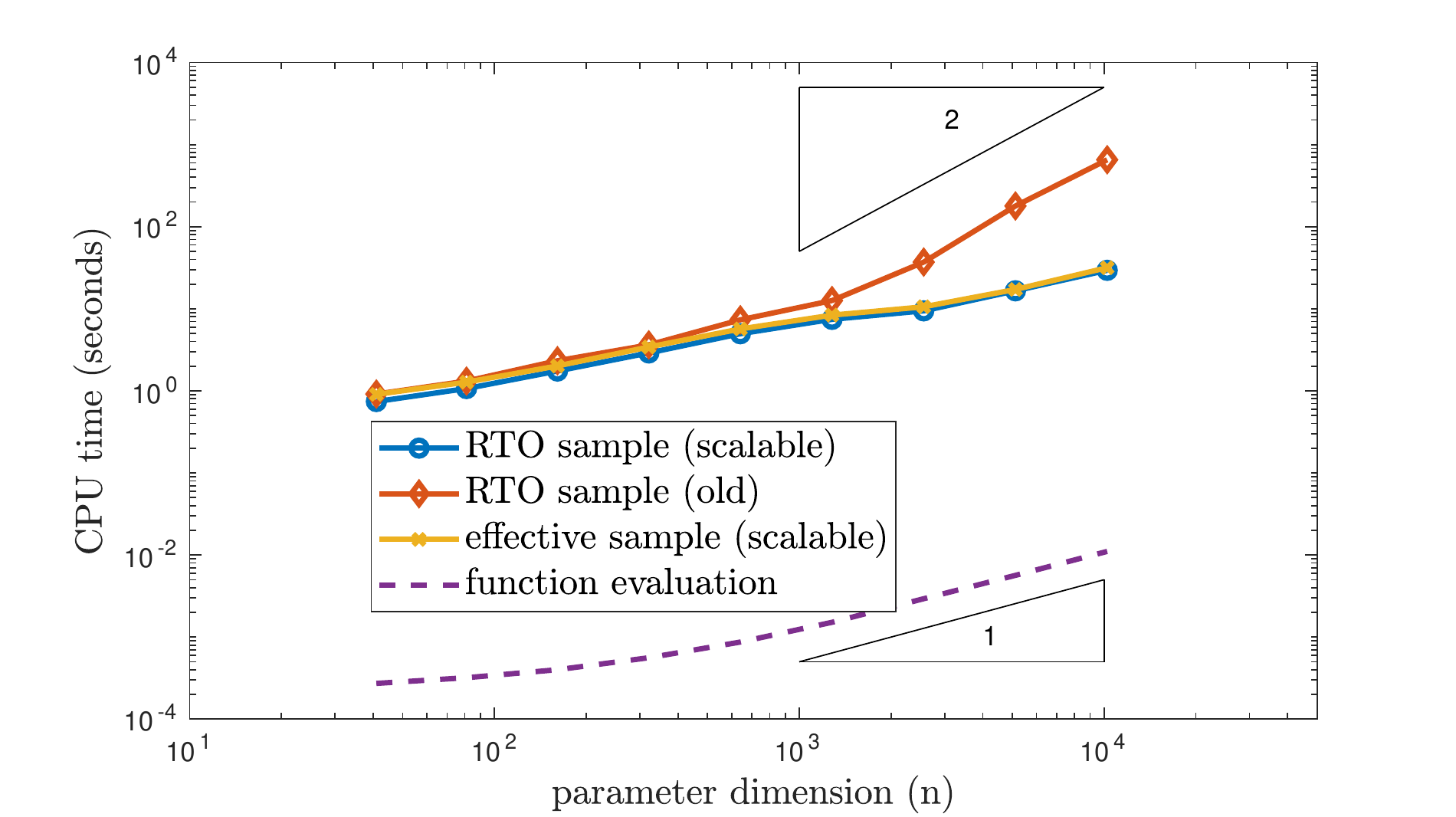}
\caption{Computational cost for elements of RTO, varying parameter dimension.}\label{fig:cpu_vs_dim}
\end{figure}

\subsection{Influence of observational noise} \label{sec:rto_sig}
In our second experiment, we examine the effect of observational noise magnitude on the sampling efficiency of RTO. We fix the parameter dimension to $n=641$ and scan through observational noise standard deviations ranging from $10^{-7}$ to $10^0$, which correspond to signal-to-noise ratios ranging from $1.5\times 10^7$ to $1.5$. Once again we run MCMC chains of length $5000$. Changing the observational noise magnitude changes the posterior distribution, as shown in Figure~\ref{fig:rto_post_obs}. With extremely small observational noise, the probability mass of the posterior concentrates on the manifold where the parameter values yield outputs that exactly match the data. Generally, this collapse makes the posterior more difficult to simulate using most MCMC methods.
In the case of RTO, it makes the optimization problems harder to solve. As shown in Table~\ref{tab:rto_obs}, even though the ESS and acceptance rate remain relatively constant with varying observational noise, the number of optimization iterations required to obtain each sample increases as the observational noise becomes very small. Thus, as the observational noise shrinks, more function evaluations are required for each MCMC step. This behavior is also illustrated in Figure~\ref{fig:cpu_vs_obs}, where the CPU time for a single function evaluation is constant, but the time for one MCMC step and for one independent sample increases. 

Of course, the number of optimization iterations at each step depends on the choice of stopping tolerance. In these experiments, we fix the function tolerance (see \S\ref{sec:rto_dim}) to $10^{-6}$. However, the forward model and the observed data enter the RTO objective function through the whitening transform (cf. Section \ref{sec:target}). This implicitly normalizes the observational noise by the standard deviation. This way, a fixed tolerance implicitly imposes an increasingly stringent condition for smaller observational noise, which explains the higher number of function evaluations required. 
Overall, though, these results suggest that RTO can be applied to inverse problems with extremely small observational noise provided that solving the optimization problems remains tractable.

\begin{figure}[ht]
\centering
\begin{subfigure}{0.49\textwidth}
\centering
\pic{1}{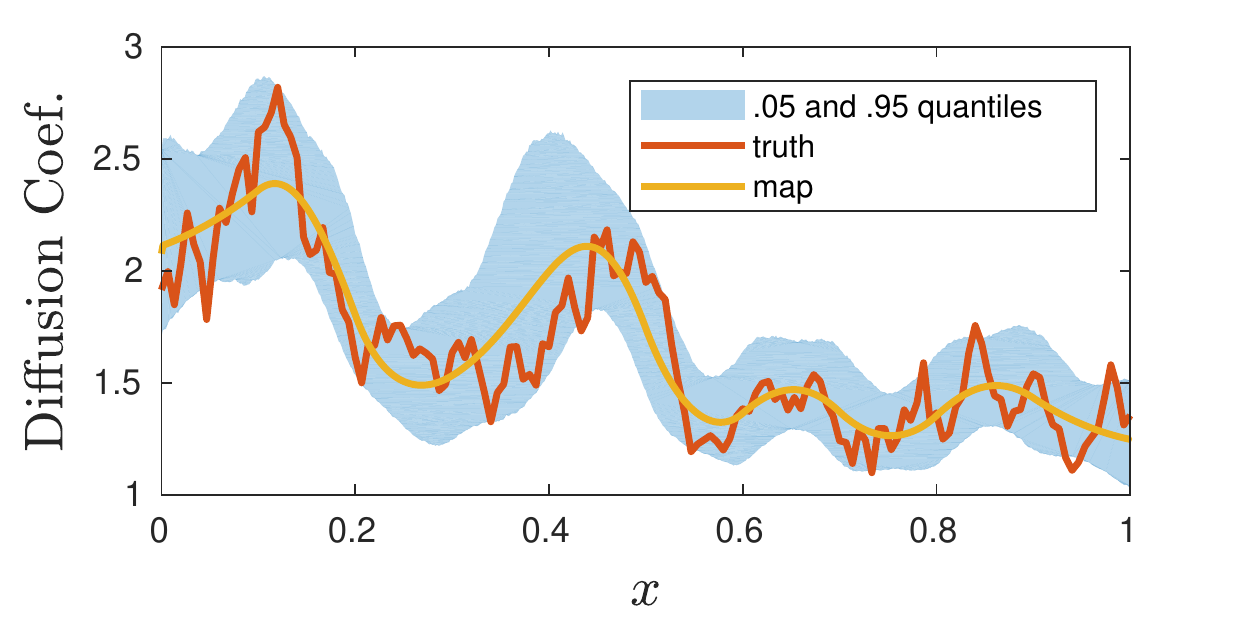}
\caption{$\sigma = 10^{-6}$}
\end{subfigure}
\begin{subfigure}{0.49\textwidth}
\centering
\pic{1}{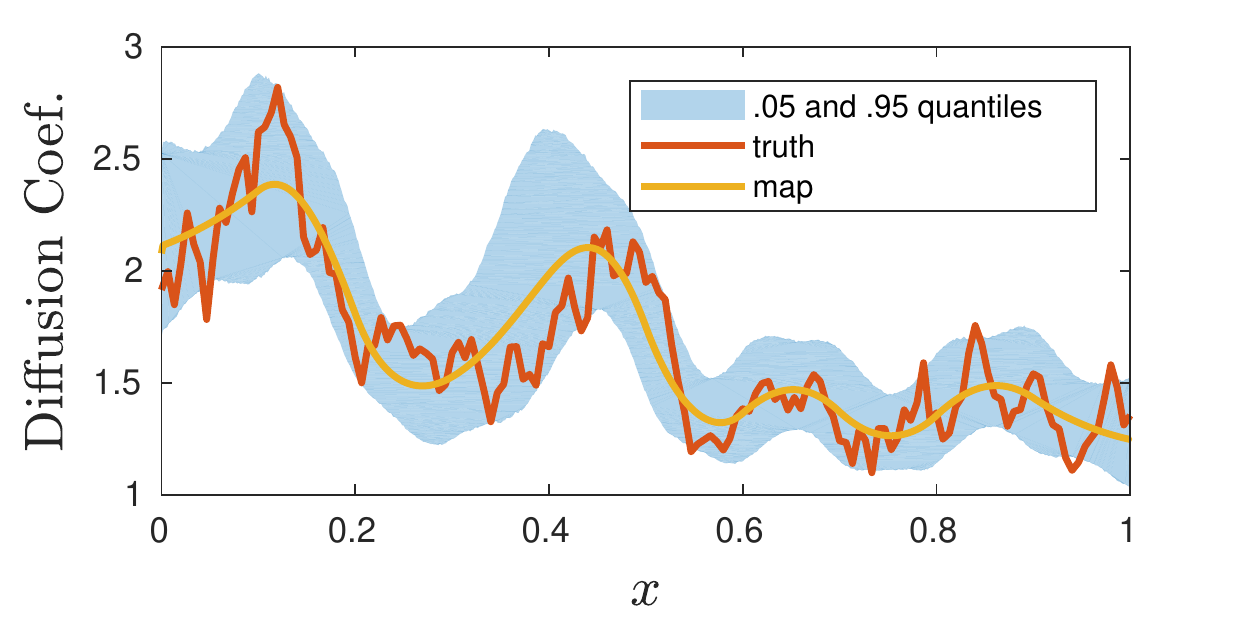}
\caption{$\sigma = 10^{-4}$}
\end{subfigure}
\begin{subfigure}{0.49\textwidth}
\centering
\pic{1}{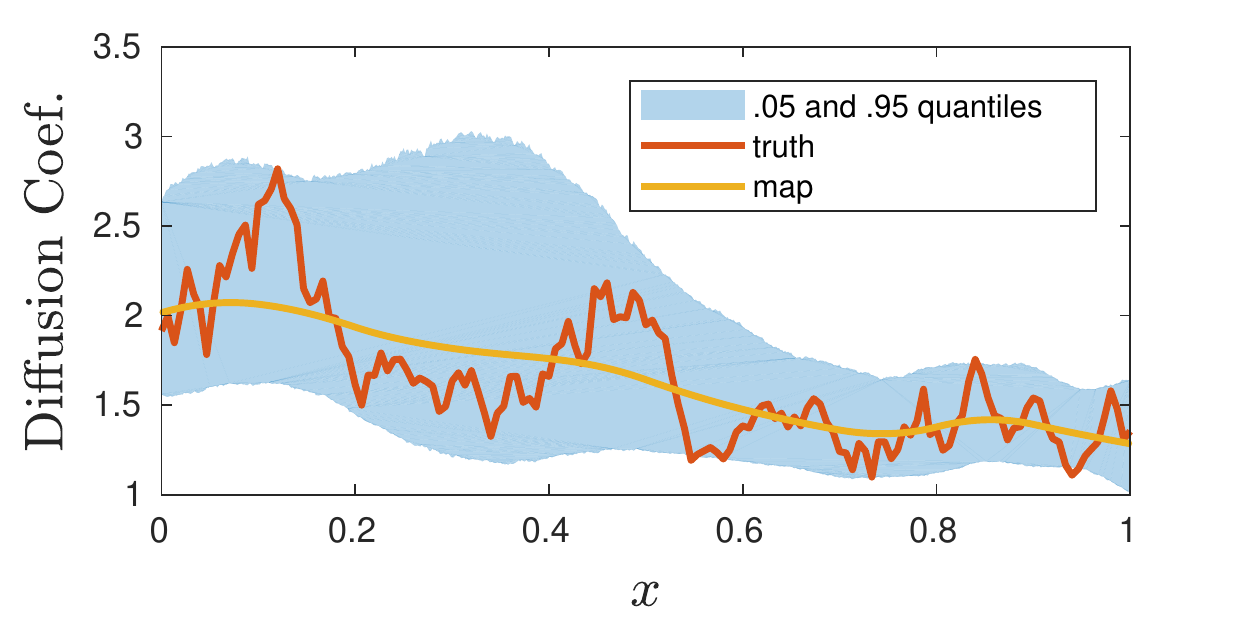}
\caption{$\sigma = 10^{-2}$}
\end{subfigure}
\begin{subfigure}{0.49\textwidth}
\centering
\pic{1}{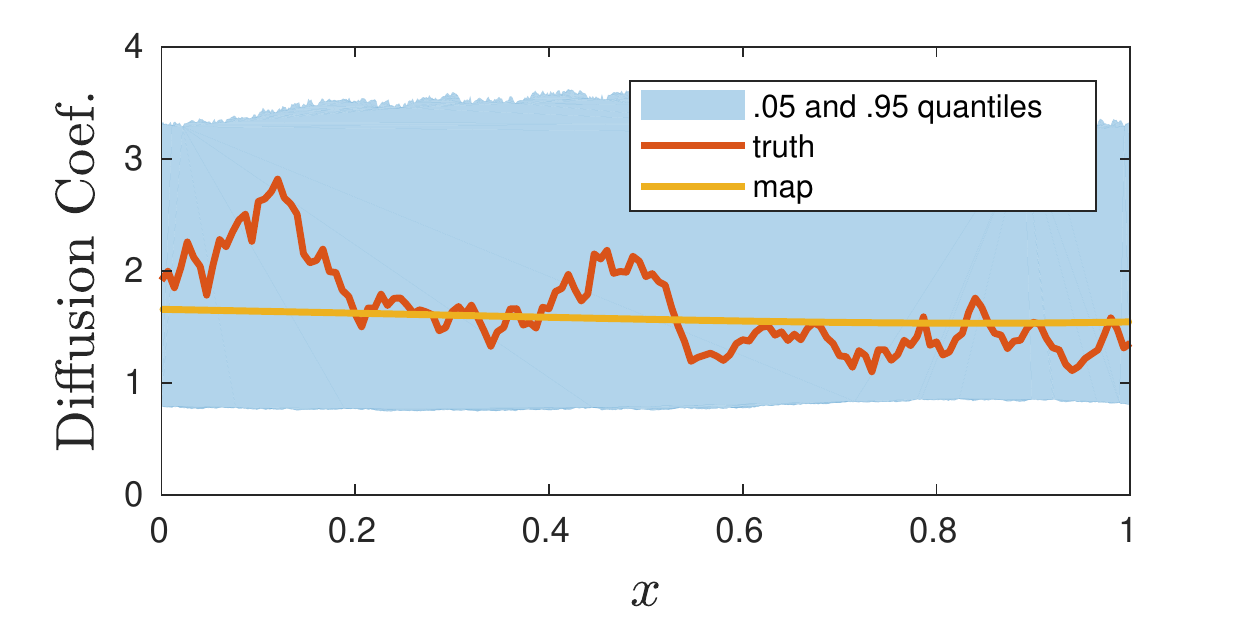}
\caption{$\sigma = 10^{0}$}
\end{subfigure}
\caption{Summary statistics of posterior distributions computed via RTO-MH with varying observational noise $\sigma$. 90\% marginal credibility intervals (blue shaded region), true diffusivity coefficient (red line), and MAP estimate (yellow line).}\label{fig:rto_post_obs}
\end{figure}

\begin{table}[htbp]
{\centering
\caption{Effective sample size (ESS), average acceptance rate, and average number of optimization iterations per step for RTO, for varying observational noise magnitude. Chain length of $5000$.}\label{tab:rto_obs}
\small
\begin{tabular}{l r r r r r r r r r}
\toprule
Noise std deviation & $10^{-7}$ & $10^{-6}$ & $10^{-5}$ & $10^{-4}$ & $10^{-3}$ & $10^{-2}$ & $10^{-1}$ & $10^{0}$ & $10^{1}$ \\
\midrule
Numerical ESS & 4504.8 & 4427.4 & 4349.9 & 4423.0 & 4415.1 & 4187.2 & 4317.7 & 4476.9 & 5000.0 \\
Acceptance Rate & 0.946 & 0.944 & 0.941 & 0.945 & 0.935 & 0.924 & 0.939 & 0.959 & 0.999 \\
Opt.\ Iterations & 567.64 & 495.41 & 363.71 & 296.55 & 89.07 & 8.32 & 5.70 & 4.70 & 3.31 \\
\bottomrule
\end{tabular}\\
}%
\end{table}

\begin{figure}[ht]
\centering
\pic{0.7}{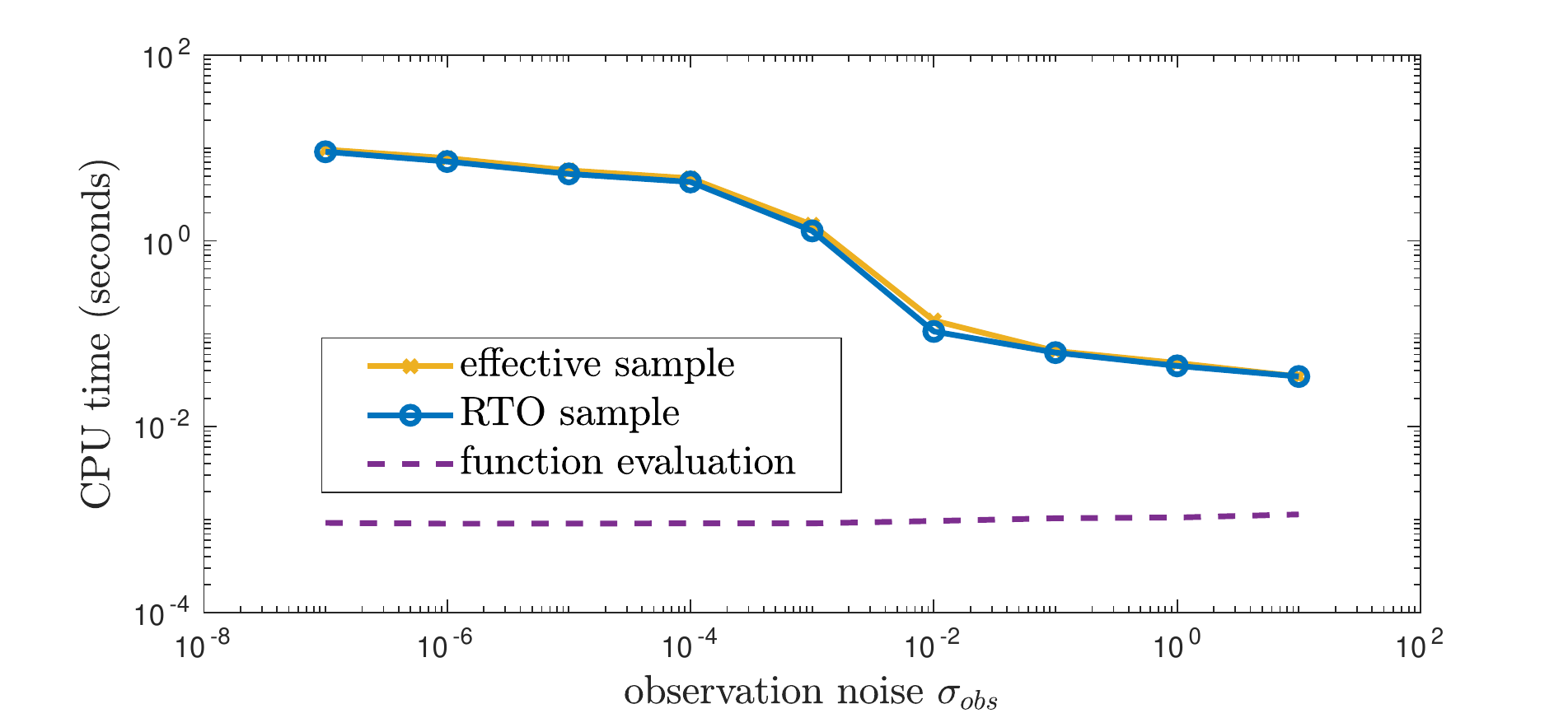}
\caption{Computational cost for elements of RTO, varying observational noise.}\label{fig:cpu_vs_obs}
\end{figure}

\subsection{Comparing RTO with pCN} \label{sec:rto_pcn}
In our third experiment, we compare the computational efficiency of RTO and pCN \cite{pCN}. The two algorithms are both dimension-independent. We fix the parameter dimension to $n=641$ and compare the algorithms' performance on inverse problems with different observational noise standard deviations, ranging from $10^{-6}$ to $10^{0}$. %
For pCN, we use a chain length of $5 \times 10^6$ and remove the first 50\% of the samples as burn-in. We manually tune the step size of pCN to obtain the largest empirical ESS. As shown in Figure~\ref{fig:pcn_post_obs}, the posterior marginals from pCN match those obtained with RTO for the two larger observational noise values. For the two smaller observational noise values, however, pCN does not converge. In particular, examination of  Figure~\ref{fig:pcn_post_obs} and of MCMC trace plots for the smaller noise cases shows that the pCN chain does not travel far from its starting point. Table~\ref{tab:compare} reveals that RTO requires less computational time per independent sample in \emph{all} cases, even when the observational noise is larger. (Note that this performance metric, time per ESS, normalizes away the impact of different chain lengths.)

In this numerical example, RTO thus outperforms pCN by a large margin. Moreover, in the two cases with smaller observational noise, RTO is the only algorithm that produces meaningful estimates of the posterior. In summary, we find that RTO's sampling performance is robust to parameter dimension and observational noise, and can be more efficient than pCN. 
\begin{figure}[ht]
\centering
\begin{subfigure}{0.49\textwidth}
\centering
\pic{1}{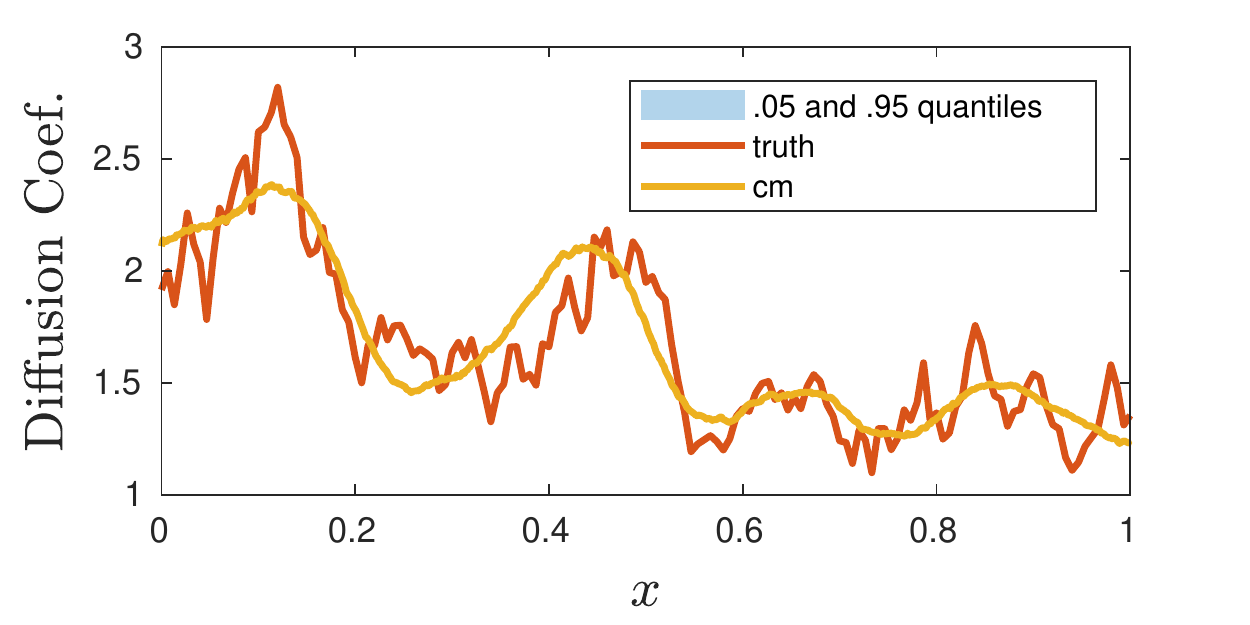}
\caption{$\sigma = 10^{-6}$}
\end{subfigure}
\begin{subfigure}{0.49\textwidth}
\centering
\pic{1}{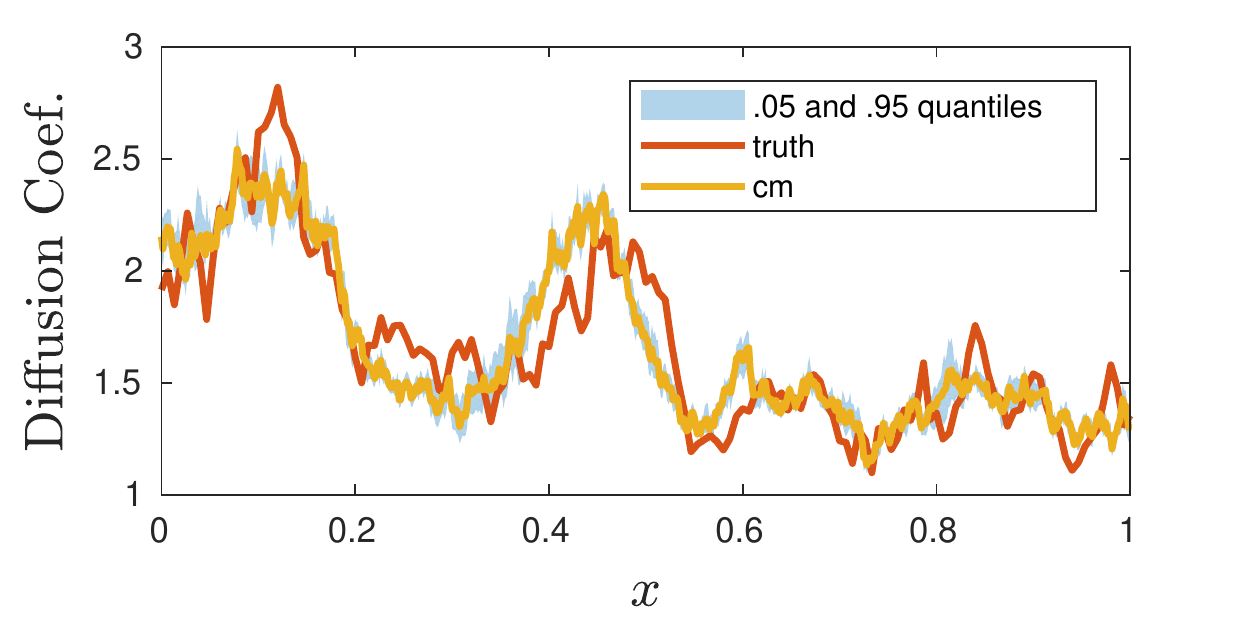}
\caption{$\sigma = 10^{-4}$}
\end{subfigure}
\begin{subfigure}{0.49\textwidth}
\centering
\pic{1}{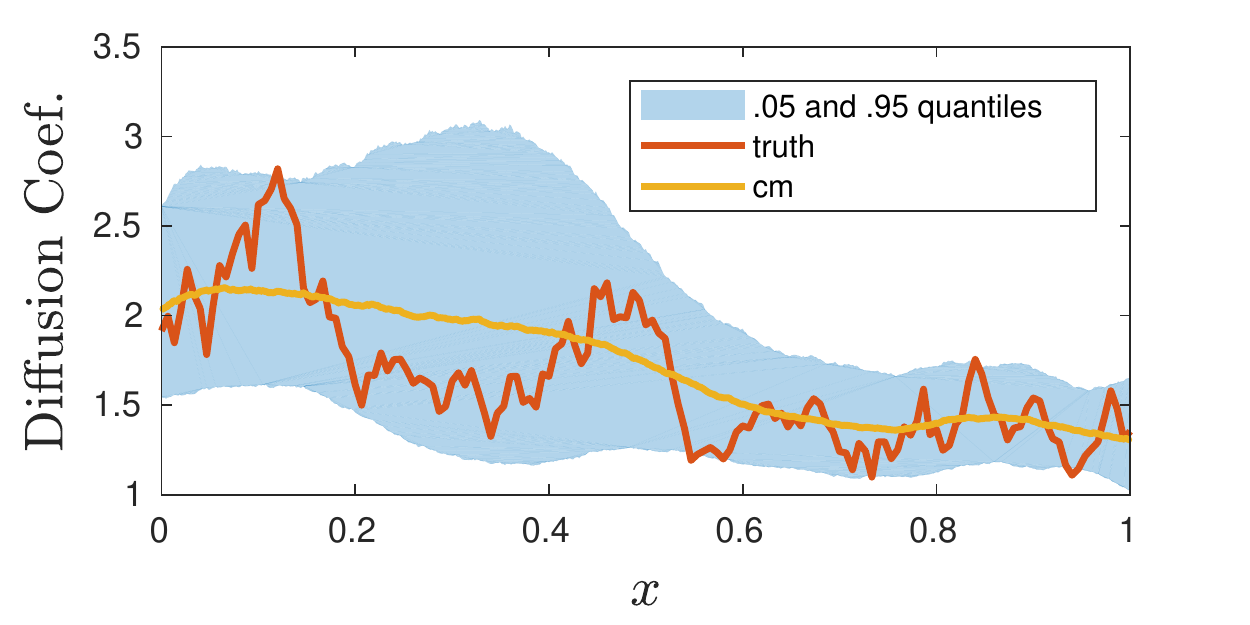}
\caption{$\sigma = 10^{-2}$}
\end{subfigure}
\begin{subfigure}{0.49\textwidth}
\centering
\pic{1}{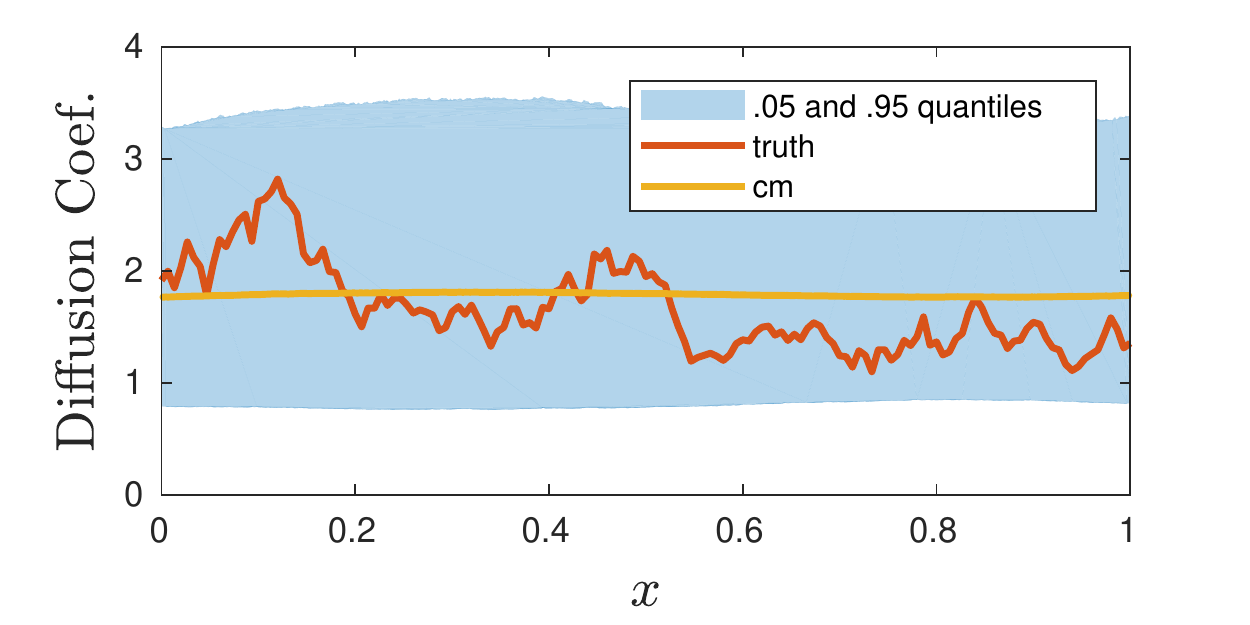}
\caption{$\sigma = 10^{0}$}
\end{subfigure}
\caption{Summary statistics of posterior distributions computed through pCN, varying observational noise $\sigma$. 90\% credibility intervals (blue shaded region), true diffusivity coefficient (red line) and CM estimate (yellow line). The MCMC chain does not converge for $\sigma = 10^{-6}$ and $\sigma = 10^{-4}$.}\label{fig:pcn_post_obs}
\end{figure}

\begin{table}[htbp]
{\centering
\caption{Comparing computational cost for RTO and pCN.}\label{tab:compare}
\small
\begin{tabular}{r r r}
\toprule
& \multicolumn{2}{c}{\textbf{CPU time (seconds) per ESS}}\\
Observational Noise $\sigma$ & \hspace{3 em} RTO & pCN \\
\midrule
$10^{-6}$ & $7.772$ & $1.193 \cdot 10^{3 *}$\\
$10^{-4}$ & $4.712$ & $1.103 \cdot 10^{3 *}$\\
$10^{-2}$ & $0.139$ & $7.739$\\
$10^0$ & $0.049$ & $0.250$\\
\bottomrule
\end{tabular}\\[3pt]
 \hspace{4 em}$^*$\small Estimated from a non-converged MCMC chain. Actual values may be higher.
}%
\end{table}

\section{Example 2: 2D parabolic PDE} \label{sec:numerics_heat}

To further demonstrate the efficacy of RTO, we solve the inverse problem of identifying the coefficient of a two-dimensional parabolic PDE from point observations of its solution.
Consider the problem domain $\Omega = [0, 3]\times [0, 1]$, with boundary $\partial \Omega$. We denote the spatial coordinate by $x = (x_1, x_2) \in \Omega$.
We model the time-varying potential (solution) field $p(x,t)$ for a given conductivity (coefficient) field $\kappa(x)$ and forcing function $f(x,t)$ using the heat equation
\begin{equation}\label{eq:heat}
\frac{\dd p(x,t)}{\dd t} = \nabla \cdot \left( \kappa(x) \nabla p(x,t) \right) + f(x,t), \quad x \in \Omega, \; t \in [0, T],
\end{equation}
where $T = 2$. Parabolic PDEs of this type are widely used in modeling groundwater flow, optical diffusion tomography, the diffusion of thermal energy, and numerous other common scenarios for inverse problems.
Let $\partial \Omega_\tx{n} = \{ x \in \partial \Omega \,|\, x_2 = 0\}  \cup  \{ x \in \partial \Omega \,|\, x_2 = 1\}$ denote the top and bottom boundaries, and $\partial \Omega_\tx{d} = \{ x \in \partial \Omega \,|\, x_1 = 0\} \cup \{ x \in \partial \Omega \,|\, x_1 = 3\}$ denote the left and right boundaries. 
For $t \geq 0$, we impose the mixed boundary condition:
\[
p(x,t) = 0, \forall x \in \partial \Omega_\tx{d}, \quad \tx{and} \quad (\kappa(x) \nabla p(x,t) ) \cdot \vec{n}(x) = 0, \forall x \in \partial \Omega_\tx{n},
\] 
where $\vec{n}(x)$ is the outward normal vector on the boundary. We also impose a  zero initial condition, i.e., $p(x,0) = 0, \forall x \in \Omega$, and let the potential field be driven by a time-invariant forcing function
\[
f(x, t) = c\,\Big( \exp\big(-\frac{1}{2 r^2} \| x - a\|^2 \big) - \exp\big(-\frac{1}{2 r^2} \| x - b\|^2 \big) \Big), \forall t \geq 0,
\]
with $r = 0.05$, which is the superposition of two Gaussian-shaped sink/source terms centered at $a = (0.5, 0.5)$ and $b = (2.5, 0.5)$, scaled by a constant $c = 6\times 10^{-4}$.

The conductivity field $\kappa(x)$ is endowed with a log-normal prior. That is, letting $u(x) = \log \kappa(x)$, the prior for $u(x)$ takes the form $\N(m_\tx{pr}, \Gamma_\tx{pr})$. Here we prescribe zero prior mean, $m_\tx{pr} = 0$, and model the inverse of the prior covariance operator using the stochastic PDE approach (see \cite{lindgren,StuartActa} and references therein):
\begin{equation}\label{eq:heat_prior}
 - \triangle u(x) + \gamma u(x) = \mathcal{W}(x), \quad x\in\Omega,
\end{equation}
where $\triangle $ is the Laplace operator and $\mathcal{W}(x)$ is the white noise process. We impose a no-flux boundary condition on the above SPDE and set $\gamma = 5$.

Equations \eqref{eq:heat} and \eqref{eq:heat_prior} are solved using the finite element method with bilinear basis functions. A mesh with $120 \times 40$ elements is used in this example. This leads to $n = 4800$ dimensional discretised parameters.
The ``true'' conductivity field used for generating observed data is a realization from the prior distribution. The true conductivity field and the simulated potential field at different times are shown in Figure \ref{fig:setup_h}(a)--(c).
The potential field is observed at $13$ discrete locations (shown as dots in Figure \ref{fig:setup_h}(a)) at $20$ discrete time points equally spaced between $t=0.1$ and $t=2$. 
We set the standard derivation of the observation noise to $\sigma = 3\times 10^{-7}$, which corresponds to a signal-to-noise ratio of about $10$.
In the inverse problem, we use this $m = 260$ dimensional vector of data to estimate the conductivity field $\kappa(x)$. 

\begin{figure}[h]
\begin{subfigure}{0.49\textwidth}
\centering
\pic{1}{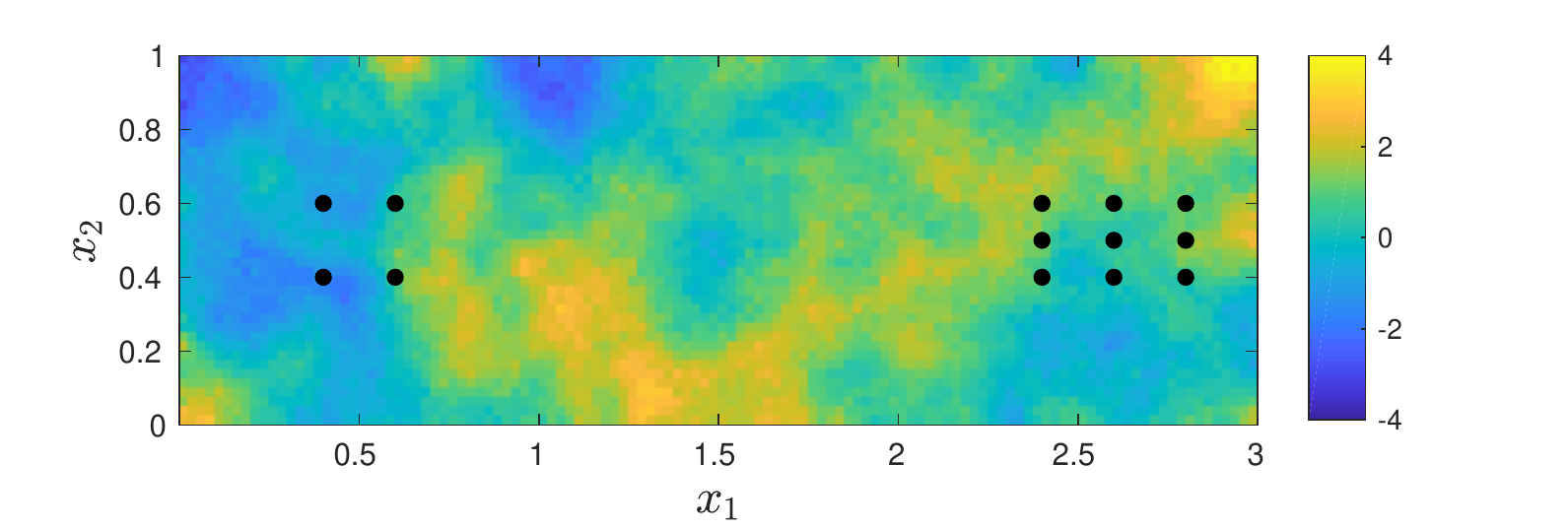}
\caption{``True'' conductivity field. The observation locations are shown as dots.\vspace{-1.3em}}
\end{subfigure}
\begin{subfigure}{0.49\textwidth}
\centering
\pic{1}{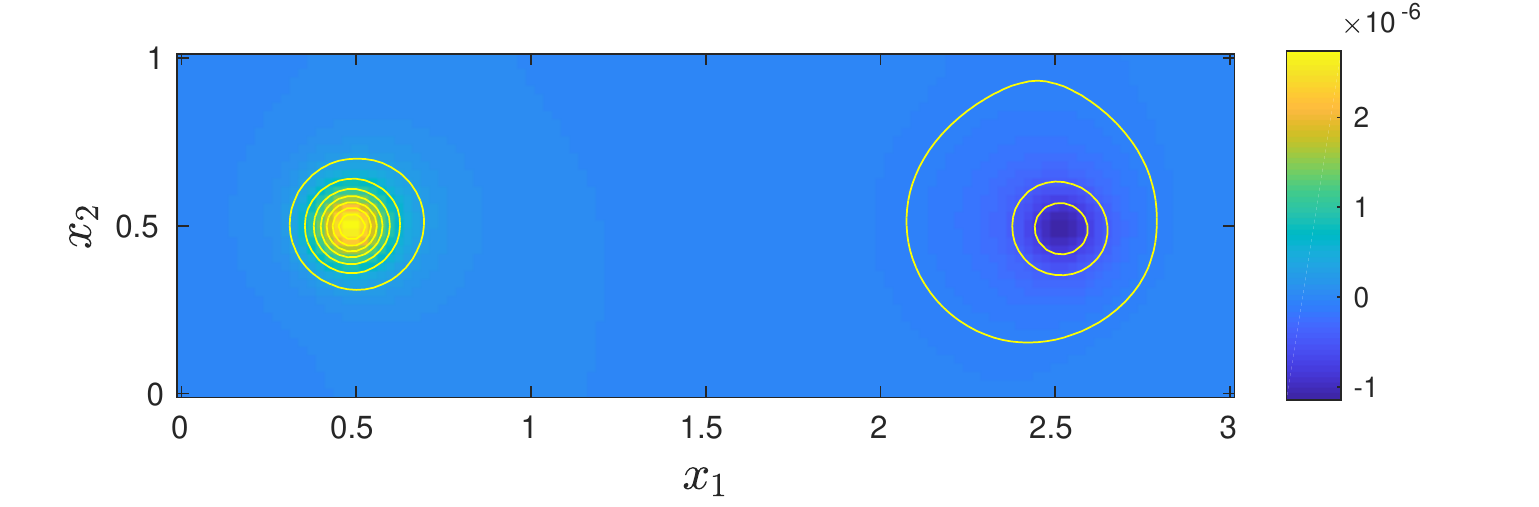}
\caption{Potential field at $t = 0.03$}
\end{subfigure}
\begin{subfigure}{0.49\textwidth}
\centering
\pic{1}{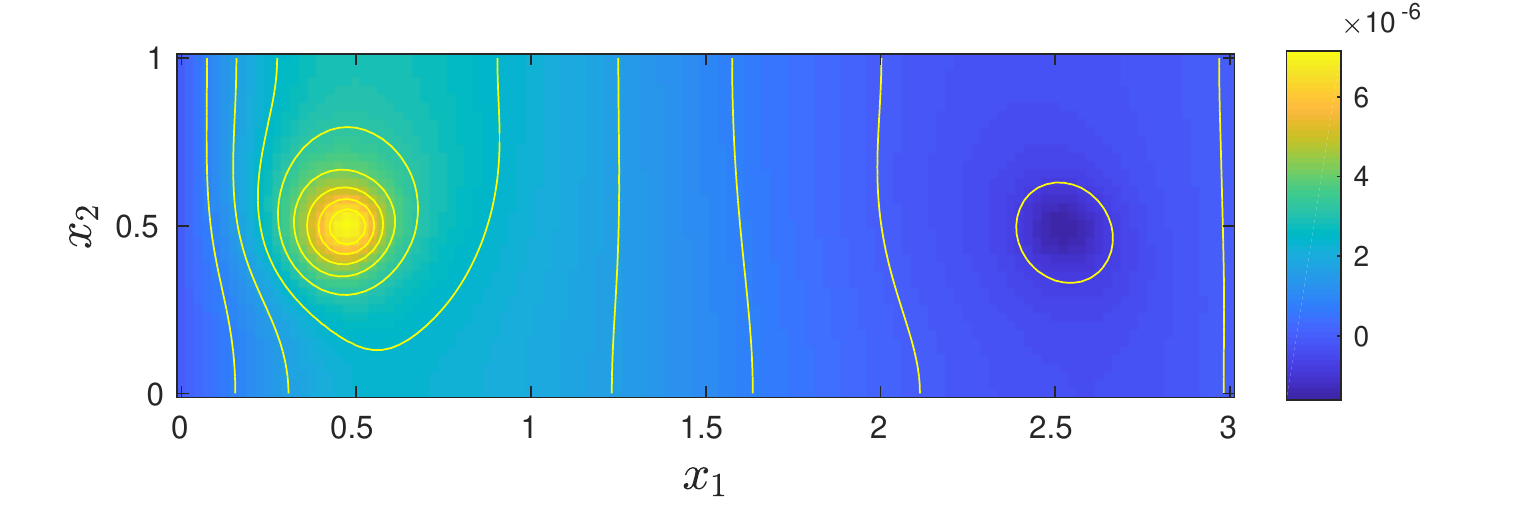}
\caption{Potential field at $t = 2.0$}
\end{subfigure}
\begin{subfigure}{0.47\textwidth}
\centering
\pic{1}{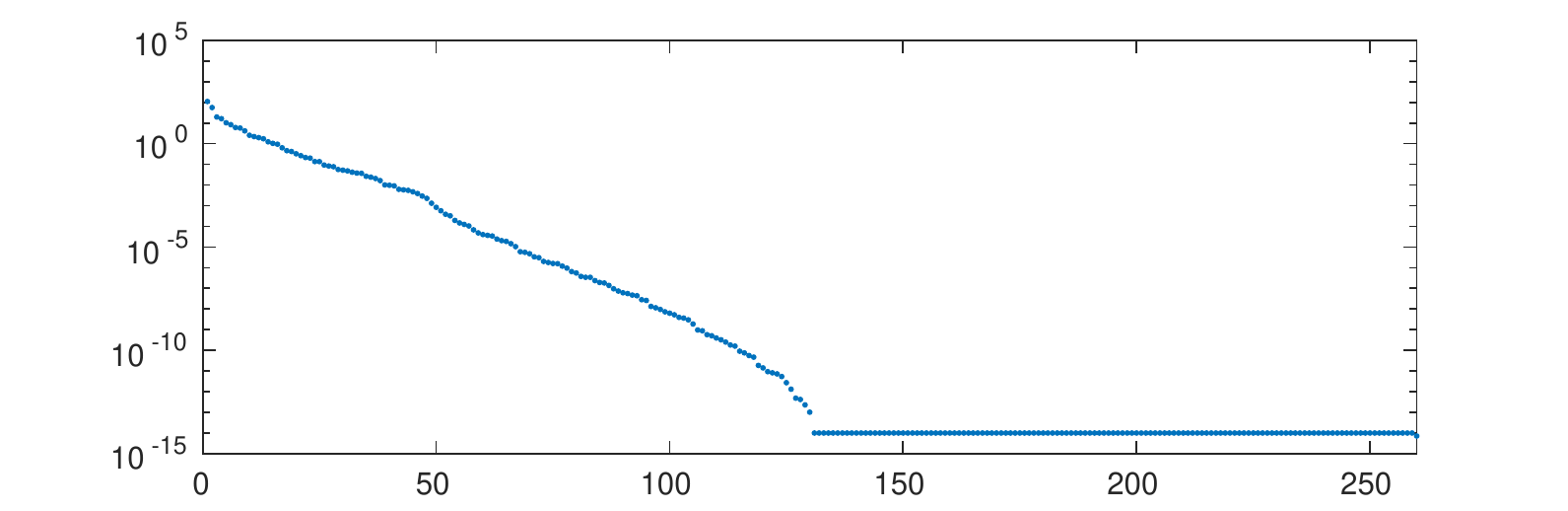}
\caption{Singular values of the linearized forward model at the MAP point.}
\end{subfigure}
\caption{Setup of the parabolic inversion example.}
\label{fig:setup_h}
\end{figure}

The forward model is linearized at the MAP point. As shown in Figure \ref{fig:setup_h}(d), we observe a sharp decay in the singular values of the linearized model, with these values dropping below machine precision after rank $130$. We truncate the singular values at thresholds $\tau = 1$, $10^{-2}$, and $10^{-4}$ to define three different RTO proposals. 
Then, using each RTO proposal, we generate $2500$ samples to characterize the posterior using a Metropolis independence sampler (i.e., RTO-MH). 
The rank of the truncated SVD, statistics about the computation of each RTO sample, and effective sample size are reported in Table \ref{tab:rto_rank} for each truncation threshold. Note that all the truncated ranks are significantly smaller than the parameter dimension $n = 4800$. 

Here, we observe that with a rather large truncation threshold ($\tau = 1$), we obtain a significantly lower ESS than with the other two truncation thresholds. This behavior agrees with the heuristics discussed in Section \ref{sec:rto_complexity}: if one truncates the SVD more aggressively, the RTO proposal gets closer to the prior. Thus, it is expected that this RTO proposal will have lower statistical performance than an RTO proposal obtained with smaller $\tau$ (e.g., $10^{-2}$). 
Once the truncation threshold is sufficiently small, however, we do not gain additional statistical performance by allowing more modes; compare the ESS at $\tau = 10^{-2}$ to that at $\tau = 10^{-4}$. This behavior is also in accordance with the truncation strategies and interpretation of the singular values discussed in Section \ref{sec:rto_complexity}.
Regarding the computational performance, we observe that more optimization iterations and longer CPU times are needed to obtain one RTO sample (on average) with the truncation threshold $\tau = 1$ than with the smaller truncation thresholds. We attribute this behavior to fact that the truncated proposal does not constrain the parameter value in directions complementary to the range of $\Phi$, and thus the optimization iterations may need to navigate through the tails of the posterior. %
For truncation thresholds $10^{-2}$ and $10^{-4}$, the difference in the number of optimization iterations is insignificant.
Overall, the truncation threshold of $\tau \approx 10^{-2}$ suggested in Section \ref{sec:rto_complexity} appears to be a reasonable choice in this example.

\begin{table}[ht]
{\centering
\caption{Rank of the truncated SVD, average number of forward model evaluations, average number of MVPs with the linearized forward model and its adjoint, average number of optimization iterations per RTO sample, average CPU time per RTO sample, and ESS; all for varying SVD truncation thresholds. Chain length of $2500$.}\label{tab:rto_rank}
\small
\begin{tabular}{l r r r }
\toprule
Truncation threshold & $1$ & $10^{-2}$ & $10^{-4}$ \\
\midrule
Rank & 15 & 39 & 57 \\
Number of evaluations of $G(v)$ & 18.8 & 12 & 12.4 \\
Number of MVPs with $\nabla G(v)$ & 321.6 & 235.6 & 258.4 \\
Optimization iterations per sample & 17.8 & 11 & 11.4 \\
CPU time (sec) per sample & 354 & 257 & 283 \\
Numerical ESS (out of 2500) & 292 & 1140 & 1130 \\
\bottomrule
\end{tabular}\\
}%
\end{table}

Two posterior samples and some summary statistics of the posterior, computed using RTO-MH with the truncation threshold $\tau = 10^{-2}$, are shown in Figure~\ref{fig:result_h}. We observe that the posterior samples and the posterior mean demonstrate similar structure to the ``true'' conductivity field used to generate the synthetic data set. We also observe that the posterior standard deviation of the conductivity field is low in regions near the observation locations. In comparison, the posterior standard deviation is relatively high in regions near the boundary and between clusters of observation locations, where  the observed data do not provide sufficient information to infer parameters.

We also attempted to compare RTO with pCN in this example. However, because of the rather informative data, pCN fails to produce an ergodic chain in a comparable amount of CPU time.
An additional, but important, implementation note is that we generated RTO samples and evaluated the corresponding weighting functions \eqref{eq:w_svd} in parallel, and then quickly postprocessed the RTO samples using the Metropolis procedure to obtain posterior samples. Postprocessing is the only serial step of the calculation, and is very fast since all the costly calculations (sample generation, weight evaluation) are already completed.
In this way, RTO can significantly reduce the wall clock time of Markov chain simulation compared to common MCMC methods that use state-dependent transition kernels, since posterior density evaluations and Markov chain simulation must be carried out sequentially in the latter case.

\begin{figure}[ht]
\begin{subfigure}{0.49\textwidth}
\centering
\pic{1}{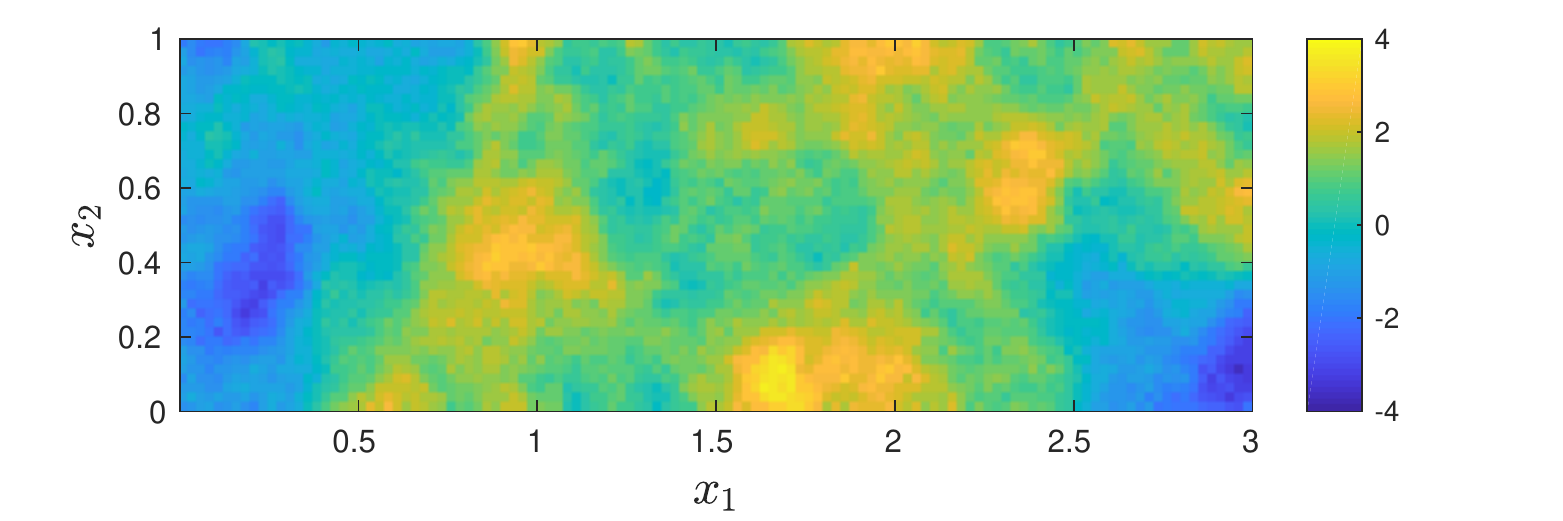}
\caption{A realization of $\kappa(x)$.}
\end{subfigure}
\begin{subfigure}{0.49\textwidth}
\centering\hspace{-1.5em}
\pic{1}{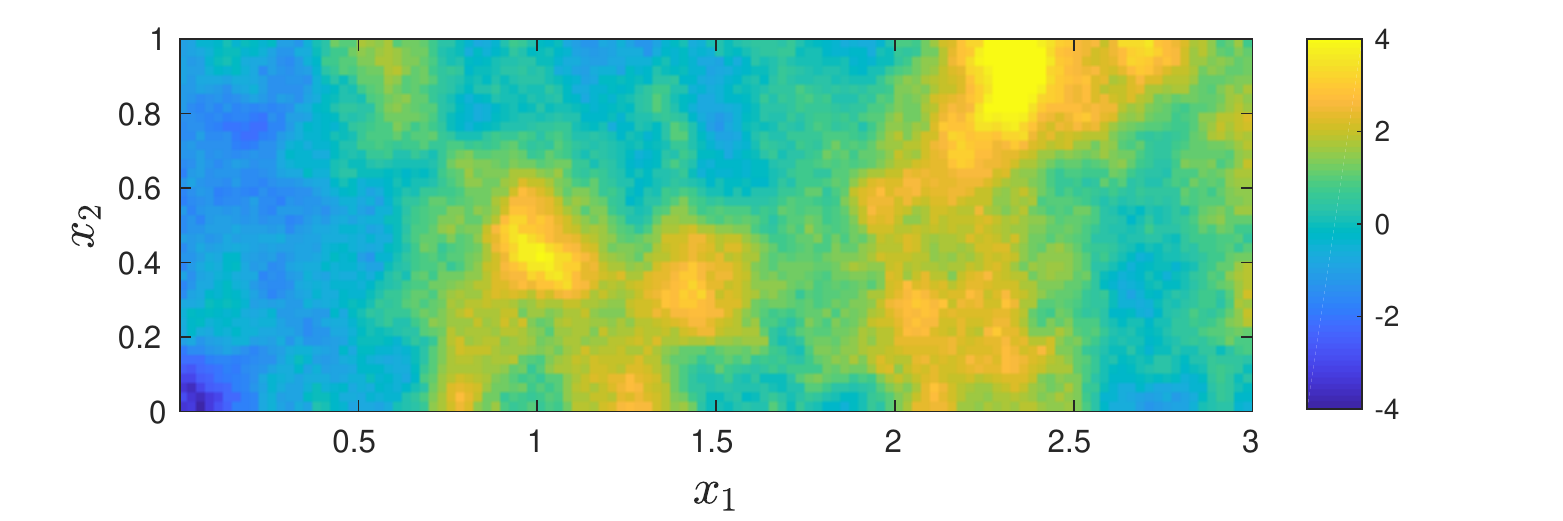}
\caption{Another realization of $\kappa(x)$.}
\end{subfigure}
\begin{subfigure}{0.49\textwidth}
\centering
\pic{1}{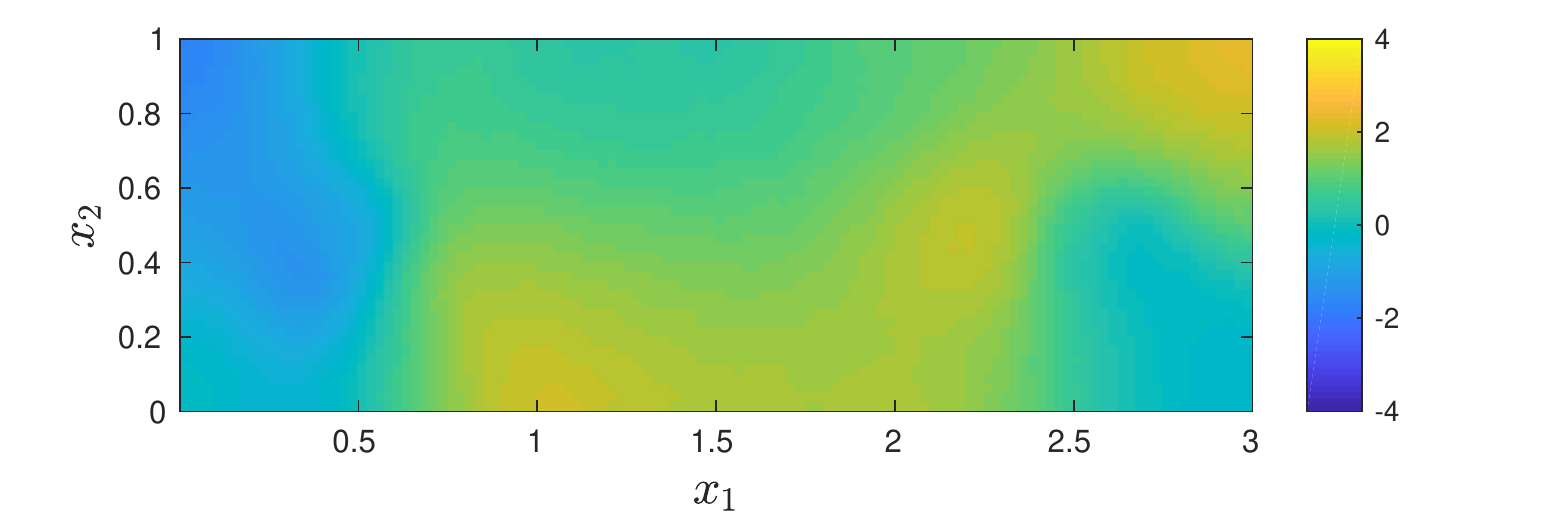}
\caption{The posterior mean of $\kappa(x)$.}
\end{subfigure}
\begin{subfigure}{0.49\textwidth}
\centering
\pic{1}{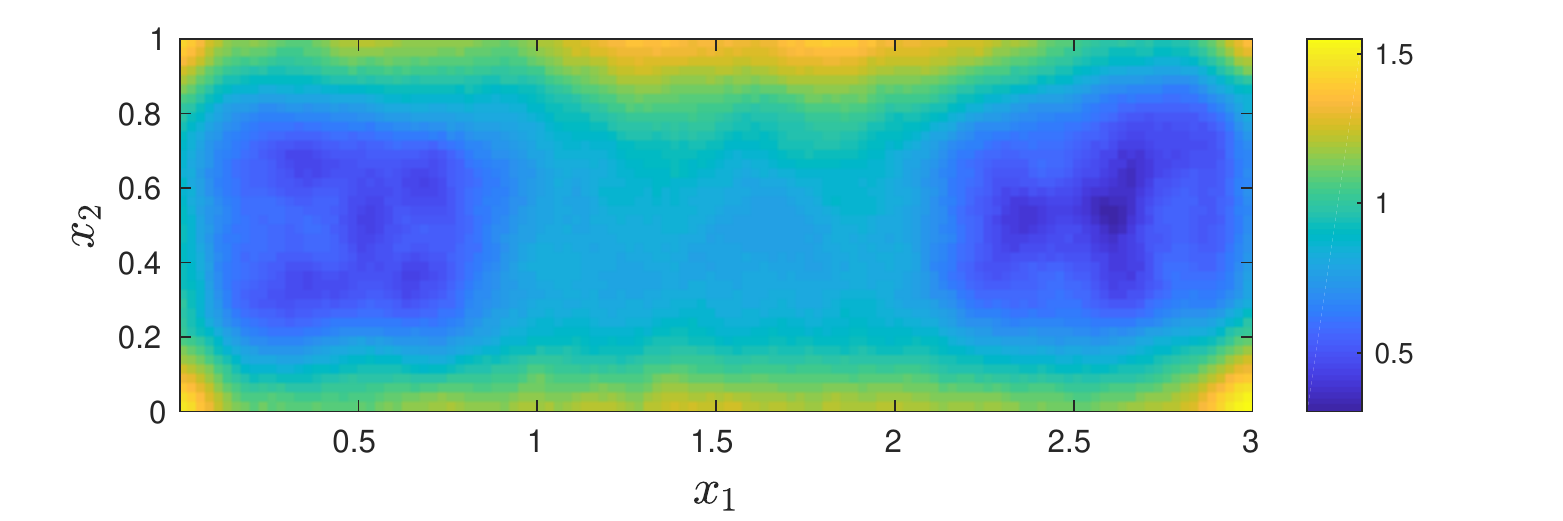}
\caption{The posterior standard deviation of $\kappa(x)$.}
\end{subfigure}
\caption{Sample realizations and summary statistics of the conductivity field $\kappa(x)$ distributed according to the posterior. }
\label{fig:result_h}
\end{figure}

\section{Discussion} \label{sec:conclusion}
The main contribution of this work is a new scalable implementation of the RTO optimization-based sampling method. 
By using a polar decomposition rather than a QR factorization to build the RTO proposal, and deriving this polar decomposition from the SVD of a linearized forward model, we can reduce the computational cost of evaluating the RTO proposal (excepting perhaps the evaluation of the forward model itself) to linear complexity in the parameter dimension.
This approach naturally splits the parameter space into two subspaces, and allows us to sample the RTO proposal and evaluate its density by solving smaller problems of size $r$, where $r$ is an intrinsic dimension of the problem. This splitting also relates the RTO proposal to other parameter dimension reduction methods for Bayesian inverse problems.
We formalize this RTO procedure in a function space setting, and show that the statistical performance of RTO is invariant to the discretized parameter dimension, under appropriate technical assumptions. 
Our results provide both practical algorithms and theoretical justification for applying RTO to high-dimensional inverse problems. 

We then provide an empirical exploration of factors influencing the sampling efficiency of RTO, using various PDE-constrained Bayesian inverse problems. Our numerical results confirm that RTO has dimension-independent sampling efficiency, and also show that the observational noise magnitude affects the cost of solving each optimization problem but not the mixing of the RTO Metropolis independence sampler. Using a simple elliptic PDE example, we observe that RTO outperforms pCN for wide range of problem settings. 
We also demonstrate the efficacy of RTO on a challenging two-dimensional parabolic PDE inverse problem, evaluating the impact of rank truncation on sampling efficiency and computational costs.
These numerical results confirm our theoretical findings: RTO offers a viable way to tackle inverse problems with high-dimensional parameters and even very small observational noise.

There are many ways to extend the work described here.
For example, once might use a mixture of several RTO proposals, defined by different linearizations, to better capture forward model nonlinearity in some extremely challenging inverse problems. Such mixtures might also help surmount the invertibility issues that arise when the assumptions of Theorem \ref{theo:rtoeqpr} are violated. For instance, one could employ a defensive mixture involving the prior distribution, along with localized proposals that are managed with trust-region strategies. 
The transport-map interpretation of RTO also suggests combining the RTO map with more elaborate local MCMC proposals on the Gaussian reference space, along the lines of~\cite{mattMCMC}. 
In addition, since RTO's prior-to-proposal mapping has a well-defined continuous limit, one can naturally use RTO to generate coupled proposal samples at different discretization levels. These correlated samples can be used as control variates in the multi-level/multi-fidelity setting \cite{multilevelpath,multilevel,multifidelity} to further accelerate the computation of posterior statistics.

\section*{Acknowledgments}
The authors would like to thank the anonymous referees for their valuable comments on improving the manuscript. We also thank Benjamin Zhang for his thoughtful comments and suggestions.
J.~Bardsley acknowledges support from the Gordon Preston Fellowship offered by the School of Mathematics at Monash University.
T.~Cui acknowledges support from the Australian Research Council, under grant number CE140100049 (ACEMS).
Y.\ Marzouk and Z.\ Wang acknowledge support from the United States Department of Energy, Office of Advanced Scientific Computing Research, AEOLUS Mathematical Multifaceted Integrated Capability Center.

\appendix

\section{Other optimization-based samplers}
\label{sec:app_imp}
Similar to RTO, other optimization-based sampling algorithms such as the random-map implementation of implicit sampling \cite{implicit} and Metropolized RML \cite{auRML} also yield deterministic couplings of two random variables.
Here we briefly review the transport maps defined by the random-map implementation of implicit sampling and by Metropolized RML.

Implicit sampling requires that the target density have level sets that are ``star-shaped,'' in that any ray starting from the mode passes through each level set exactly once. The target density is written as
\begin{equation} 
\pi_\tx{tar}(v) \propto \exp{\lb( -\ell(v) \rb)}, 
\end{equation} 
where the negative log-target density $\ell$ has a minimum at the mode $v_\tx{MAP}$. In order to draw proposal samples, we sample $\xi \in \R^n$ from a standard Gaussian and solve the following nonlinear system of equations to find a proposal $v_\ast \in \R^n$ :
\begin{equation} 
\lb\{\begin{aligned}
 \frac{L^{-1}(v_\ast-v_\tx{MAP})}{\|L^{-1}(v_\ast-v_\tx{MAP})\|} &= \frac{\xi}{\|\xi\|} \\
 \ell(v_\ast)-\ell(v_\tx{MAP}) &= \frac12 \|\xi\|^2 
\end{aligned}\rb. .
\end{equation} 
The \textit{direction} of the sample $v_\ast$ (relative to the mode) is based on the direction of the sampled $\xi$. The \textit{magnitude} of $v_\ast$ is then found through a one-dimensional line search for the point where the negative log target $\ell$ satisfies 
\[\ell(v_\ast) - \ell(v_\tx{MAP}) = \frac12 \|\xi\|^2 .\] 
In practice, $L$ is chosen to be a square matrix such that $L^\T L \defeq \lb[\Hess \ell(v_\tx{MAP})\rb]^{-1} $, where $\Hess \ell(v_\tx{MAP})$ is the Hessian of $\ell$ evaluated at the MAP point.

Similar to RTO, Metropolized RML requires that the target distribution have a Gaussian prior and additive Gaussian observational noise. Following the notation in Section \ref{sec:target}, we present a whitened version of Metropolized RML where the prior and observational noise covariances are transformed to the identity and the data is shifted to the origin. 
This way, the target density takes the form 
\begin{equation} 
\pi_\tx{tar}(v) \propto \exp\lb( -\frac12 \|v\|^2 -\frac12 \|G(v)\|^2 \rb).
\end{equation} 
Defining a tuning parameter $\gamma \in (0,1)$, Metropolized RML adds the auxiliary variables $d \in \R^m$ and considers an \textit{augmented} target distribution 
\begin{equation} 
\pi_\tx{tar}(v, d) \propto \exp\lb( -\frac12 \|v\|^2 -\frac{1}{2\gamma} \|G(v) - d\|^2 -\frac{1}{2(1-\gamma)} \| d \|^2\rb) 
\end{equation}
This defines a distribution on the joint space of parameters and data. Since the above joint distribution can also be written as
\begin{align*}
\pi_\tx{tar}(v, d) & \propto \exp\lb( -\frac12 \|v\|^2 -\frac12 \|G(v)\|^2 -\frac1{2\gamma(1-\gamma)} \|d - (1-\gamma) G(v) \|^2\rb) \\
& \propto \pi_\tx{tar}(v) \, \exp\lb(-\frac1{2\gamma(1-\gamma)} \|d - (1-\gamma) G(v) \|^2\rb) ,
\end{align*}
it can be expressed as product of the marginal distribution of $v$---which is the original target distribution---and the conditional distribution of $d$ given $v$.
Defining another tuning parameter $\rho \in (0,1)$, Metropolized RML generates a pair of random variables $\xi_v \sim \N(0,\I_{n})$ and $\xi_d \sim \N(0,\I_{m})$ and solve the following randomly perturbed optimization problem
\begin{equation*} 
(v_\ast, d_\ast) = \argmin_{(v, d)} \lb( \frac12 \|v - \xi_v\|^2 +\frac{1}{2\rho} \|G(v) - d\|^2 + \frac{1}{2(1-\rho)} \| d  - \xi_d\|^2 \rb),
\end{equation*}
to obtain a pair of proposal samples $(v_\ast, d_\ast)$.
Under the first order optimality condition, at the minima of the above objective function, the following system of nonlinear equations holds:
\begin{equation}\label{eq:rml_map}
\lb\{\begin{aligned}
v_\ast + \frac1\rho \nabla G(v_\ast)^\T \lb(G(v_\ast) - d_\ast\rb) &= \xi_v  \\
\frac1\rho d_\ast -\lb(\frac{1-\rho}\rho\rb) G(v_\ast) &= \xi_d 
\end{aligned}\rb. .
\end{equation}
Thus, one can compute the joint density of $(v_\ast, d_\ast)$ in the augmented parameter-and-data space using the mapping defined in \eqref{eq:rml_map}. 
The samples are Metropolized in the augmented space to obtain correlated samples distributed according the augmented target distribution $\pi_\tx{tar}(v, d)$. The components of $v$ are then distributed according to the original target distribution. The parameters $\rho$ and $\gamma$ are tunable settings of the algorithm. In practice, they are set close to one and zero respectively. 

A summary of the mappings induced by RTO, implicit sampling, and RML is given in Table~\ref{tab:maps}. Each algorithm describes a different map $S$, as in \eqref{eq:RTOmap}, to build the deterministic coupling. The actions of the inverse maps need to be computed using either nonlinear optimization algorithms or root finding methods (in one dimension). 

\begin{table}[h]
\newcommand{\specialcell}[2]{\begin{tabular}[#1]{@{}c@{}}#2\end{tabular}}

{\centering
\caption{Transport map interpretation of the three optimization-based samplers. In RTO, with default settings, the matrix $Q$ comes from a thin QR factorization. 
}\label{tab:maps}
\small
\begin{tabular}{l l l l}
\toprule
Algorithm \hspace{-1.5em}~& \pbox{5 cm}{Target distribution}  & \pbox{5 cm}{Transport map}  \\
\midrule
 RTO & $\pi_\tx{tar}(v) \propto \exp{\lb(-\frac12 \|H(v)\|^2\rb)}$ &  $Q^\T H(v) = \xi$ \\[0.5em]
 & &  $\tx{where\ }Q R \defeq \Jac H (v_\tx{ref})$ \\
\midrule
 \pbox{2.5cm}{Implicit sampling} & $\pi_\tx{tar}(v) \propto \exp{\lb( -\ell(v) \rb)}$ &
 \pbox{6cm}{  $\lb\{\begin{aligned}
 \frac{L^{-1}(v-v_\tx{ref})}{\|L^{-1}(v-v_\tx{ref})\|} &= \frac{\xi}{\|\xi\|} \\
 \ell(v)-\ell(v_\tx{ref}) &= \frac12 \|\xi\|^2 
\end{aligned}\rb.$}  \\[2.5em]
&& $\tx{where\ }L^\T L \defeq \lb[\Hess \ell(v_\tx{ref})\rb]^{-1}$ \\
\midrule
 \pbox{2.5cm}{RML \\} & 
 \pbox{6cm}{$\pi_\tx{tar}(v,d) \propto \exp\big( -\frac12 \|v\|^2 $ \\$ -\frac{1}{2\gamma} \|G(v) - d\|^2 -\frac{1}{2(1-\gamma)} \| d \|^2\big) $} & 
 \pbox{6cm}{$\lb\{\begin{aligned}
v + \frac1\rho \nabla G(v)^\T \lb(G(v) - d\rb) &= \xi_v  \\
\frac1\rho d -\lb(\frac{1-\rho}\rho\rb) G(v) &= \xi_d 
\end{aligned}\rb.$} \\[2.5em]
& $\tx{where\ }\gamma\in(0,1)$ & $\tx{where\ } \rho\in(0,1)$\\
\bottomrule
\end{tabular}\\
}%
\end{table}

\end{document}